\documentclass[lettersize,journal]{IEEEtran}

\IEEEoverridecommandlockouts                              

\usepackage{algorithm}
\usepackage{amsmath}
\usepackage{amssymb}
\usepackage[noend]{algorithmic}
\usepackage{graphicx}
\usepackage[usenames,dvipsnames]{xcolor}
\usepackage[font=small]{caption}
\usepackage{subfig, float, amsmath, amssymb, amsfonts, graphics, graphicx, enumitem, hhline}
\usepackage{listings,multicol}
\usepackage{lipsum}
\usepackage{hyperref}
\usepackage{dashrule}
\usepackage{dblfloatfix}
\usepackage{placeins}
\usepackage{balance}
\usepackage{diagbox}

\usepackage[T1]{fontenc}
\usepackage[utf8]{inputenc}
\usepackage{babel}
\usepackage[font=small,labelfont=bf]{caption}

\usepackage{midfloat,capt-of}

\usepackage{bbm}		

\newcommand{\mm}[1]{\mathcal{#1}}
\newcommand{\li}{[\![}
\newcommand{\ri}{]\!]}

\newcommand{\sZ}{\mm{Z}}
\newcommand{\mCZ}{\mm{CZ}}
\newcommand{\bfo}{\mathbf{1}}
\newcommand{\bfz}{\mathbf{0}}
\newcommand{\mtg}[1]{\boldsymbol{#1}}
\newcommand{\mt}[1]{\boldsymbol{#1}}

\newcommand{\cgeq}{=}

\newif\ifdraft
\drafttrue

\newcommand{\revise}[1]{{\color{black} #1}}

\usepackage{amsthm}

\newtheoremstyle{theoremdd}
  {\topsep}
  {\topsep}
  {}
  {0pt}
  {\bfseries}
  {.}
  { }
  {\thmname{#1}\thmnumber{ #2}\thmnote{ (#3)}}

\theoremstyle{theoremdd}

\newtheorem{assump}{Assumption}

\newtheorem{eg}{Example}

\newtheorem{prop}{Proposition}
\newtheorem{lem}{Lemma}
\newtheorem{rmk}{Remark}
\newtheorem{thm}{Theorem}

\title{
Efficient Backward Reachability Using the Minkowski Difference of Constrained Zonotopes
}

\author{ 
Liren Yang \hspace{1cm} Hang Zhang \hspace{1cm} Jean-Baptiste Jeannin \hspace{1cm} Necmiye Ozay
\thanks{
This article will be presented at the International Conference on Embedded Software (EMSOFT) 2022 and will appear as part of the ESWEEK-TCAD special issue.
This work is supported by the 
NSF grants ECCS-1553873 and CCF-1918123, and the ONR grant N000141812501. 

L. Yang is with the School of Artificial Intelligence and Automation, Huazhong University of Science and Technology,  Wuhan 430074, China.
 (e-mail: {\tt lirenyang@hust.edu.cn}) 

H. Zhang is with the  Department of Mechanical
Engineering, University of Wisconsin-Madison, Madison, WI 53706, USA (e-mail: {\tt\small hang.zhang@wisc.edu}).

J. B. Jeannin is with the Department of Aerospace Engineering, University of Michigan, Ann Arbor, MI 48105, USA (e-mail:   {\tt\small jeannin@umich.edu}).

N. Ozay is with the Department of Electrical Engineering and
Computer Science, University of Michigan, Ann Arbor, MI 48105, USA (e-mail:   {\tt\small necmiye@umich.edu}).
 }
}

\begin{document}
\maketitle
\thispagestyle{empty}
\pagestyle{empty}

\begin{abstract}
Backward reachability analysis is essential to synthesizing controllers that ensure the correctness of closed-loop systems. This paper is concerned with developing scalable algorithms that under-approximate the backward reachable sets, for discrete-time uncertain linear and nonlinear systems. Our algorithm sequentially linearizes the dynamics, and uses constrained zonotopes for set representation and computation. The main technical ingredient of our algorithm is an efficient way to under-approximate the Minkowski difference between a constrained zonotopic minuend and a zonotopic subtrahend, which consists of all possible values of the uncertainties and the linearization error. 
\revise{
This Minkowski difference needs to be represented as a constrained zonotope to enable subsequent computation, but, as we show, it is impossible to find a polynomial-size representation for it in polynomial time. 
Our algorithm finds a polynomial-size under-approximation in polynomial time. 
}
We further analyze the conservatism of this under-approximation technique, and show that it is exact under some conditions. Based on the developed Minkowski difference technique, we detail two backward reachable set computation algorithms to control the linearization error and incorporate nonconvex state constraints. Several examples illustrate the effectiveness of our algorithms. 
\end{abstract}


\section{Introduction}

\subsubsection{Backward Reachability Analysis}
Backward reachability analysis is concerned with finding a set of states (called the backward reachable set, BRS for short), from where a proper control strategy can steer the system's trajectories into a prescribed target region in finite time. 
The computation of BRSs is central to many control synthesis problems with reachability \cite{bertsekas1971minimax, lygeros1999controllers}, safety \cite{bertsekas1972infinite,mitchell2007comparing} or even more complex temporal logic requirements \cite{chen2018signal}, and can be used to seek critical test cases for closed-loop systems with complex controllers in the loop \cite{chou2018using,yang2021synthesis}. 
Whenever the exact computation is hard, an under-approximation can still be used
to define a conservative strategy that accomplishes the reachability task. 
For systems that exhibit modeling error or are affected by environmental uncertainties, the target region should be reached in a guaranteed manner, regardless of these uncertainties. 
This leads to
a conservative analysis and hence smaller BRSs. 
For linear systems with additive disturbances, this amounts to a Minkowski difference step in the sequential computation of BRSs (e.g., see \cite{bertsekas1971minimax}). 
For nonlinear systems, this is achieved by shrinking the target set \revise{\cite{goubault2019inner}},  \cite{schurmann2021formal}, which can be implemented by Minkowski subtracting a set that over-approximates the impact of the linearization error and disturbances.  
\revise{This shrinking step is absent in the forward reachability analysis, for which there is a sizable literature focusing on over-approximation (see \cite{althoff2021set} and the references therein). 
However, if one wants to under-approximate the BRSs under uncertain inputs by employing those  techniques developed for the forward computation (e.g., \cite{kochdumper2020computing,rwth2014under}), the shrinking step is necessary.} 


\subsubsection{Minkowski Difference}
Since the late 60s, a simple approach using support functions is known to compute the exact Minkowski difference (in halfspace representation, H-Rep for short) between a polyhedral minuend (in H-Rep) and a compact subtrahend \cite{hnyilicza1969aset}.   
For a thorough discussion on this subject, see \cite{kolmanovsky1998theory}. 
For high-dimensional polyhedra, unfortunately, H-Reps are not suitable for other operations such as affine transformation and Minkowski addition.
This is because an H-Rep's complexity may grow exponentially after these operations \cite{tiwary2008hardness}. 
For example, the off-the-shelf tool MPT3 \cite{MPT3} may return an error when computing the Minkowski addition between two 4-D polytopes. 
For applications like reachability analysis that extensively involve such operations, 
 algorithms can be made more scalable at the cost of generality, by considering a special class of polyhedra called zonotopes.
A zonotope can be expressed by its generator representation (G-Rep for short), which is more suitable for affine transformation and Minkowski addition. 
The Minkowski difference, however, is not as easy to compute when the minuend is in G-Rep.
Compared to other operations, the problem of Minkowski-subtracting a set from a zonotopic minuend (in G-Rep) receives less attention, and is first studied in \cite{althoff2015computing}, where the subtrahend is also assumed to be a zonotope (in G-Rep). 
\revise{
The exact Minkowski difference is not necessarily a zonotope, but a zonotopic under-approximation can still be found efficiently \cite{raghuraman2020set,yang2022scalable} using the encoding techniques developed in \cite{sadraddini2019linear}.}
Based on these developments, 
a scalable backward reachability algorithm is obtained for linear systems with additive disturbances in \cite{yang2022scalable}. 

\subsubsection{Constrained Zonotopes}
To enjoy the same computational advantages as zonotopes (and their G-Reps) while achieving the generality of polyhedra, a new set representation called constrained generator representation (CG-Rep for short) is proposed in \cite{scott2016constrained}.
A set expressible by CG-Rep is called a constrained zonotope. 
Not only can affine transformation and Minkowski addition of constrained zonotopes be done easily via CG-Rep manipulation, so can intersection, under which zonotopes are not even closed.  
Moreover, all polytopes (i.e., bounded polyhedra) are expressible by CG-Rep. 
Therefore, constrained zonotopes (in CG-Reps) serve as an efficient tool for set-based control and estimation. 
They are more general than zonotopes and are particularly suitable to deal with state constraints.  
\revise{However, the Minkowski difference operation,  which is necessary for BRS computation, is difficult for constrained zonotopes. 
In fact, we show that no polynomial-time algorithm can find a polynomial-size CG-Rep of the Minkowski difference between a constrained zonotopic minuend (in CG-Rep) and a zontopic subtrahend (in G-Rep), unless P $=$ NP.
Neither there is, to the best of our knowledge, an efficient way to compute a polynomial-size  under-approximation.
This prohibits the use of constrained zonotopes for BRS computation under uncertainties because a compact representation of the BRS is essential for its efficient end uses (e.g., checking if a state belongs to the BRS and deriving the control law accordingly). }

\vspace{-3mm}
\subsubsection{Contributions}
In this paper, we use constrained zonotopes to develop scalable algorithms that under-approximate the BRSs for discrete-time nonlinear systems.
Our approach is based on sequential linearization, and the linearization error is incorporated with a Minkowski difference step. 
Our technical contributions are summarized as follows.

\emph{i)} We propose an efficient way to under-approximate the Minkowski difference between a constrained zonotopic minuend (in CG-Rep) and a zonotopic subtrahend (in G-Rep).
Our approach is optimization-based. 
We show that a na\"ive use of the encoding from \cite{sadraddini2019linear} leads to a bilinear program, but by extending the two-step approach proposed in  \cite{yang2022scalable}, an under-approximation can be found via  a linear program. The size of this linear program is polynomial in that of the minuend's and the subtrahend's representations. \revise{Our approach hence gives a polynomial-sized under-approximation in polynomial time.}

\emph{ii)} We further analyze the conservatism of this extended two-step approach. 
In particular, we show that any constrained zonotopic minuend has a ``rich'' enough CG-Rep, for which our two-step approach is exact. 
While it may be impractical to always assume such a rich CG-Rep, this result opens the direction of incrementally enriching the given CG-Rep of the minuend to improve the two-step approach's accuracy. 

\emph{iii)} Using the developed Minkowski difference  technique, we propose two methods: scaling method and splitting method for BRS computation. The scaling method can compute BRSs with convex constraints for a longer time horizon than the splitting method. In contrast, the splitting method can give larger BRSs than those obtained by scaling method for a short time horizon but have difficulties in computing BRSs for a long time horizon. However, the splitting method can deal with nonconvex constraints and expand the BRSs into different homotopy classes. Experiments show the advantages of these constrained-zonotope-based methods: they give less conservative BRSs under-approximation than those using zonotopes \cite{yang2022scalable}, especially in the presence of state constraints, and scales better than the Hamilton-Jacobi (HJB) method~\cite{bansal2017hamilton}. 






\vspace{-2mm}
\section{Notations \& Preliminaries}



We use $\bfo$ ($\bfz$, resp.) to represent a matrix of proper size whose entries are all ones (zeros, resp.). We will not make the size of such a matrix explicit unless it is not clear from context. 
Let $\mt{M}$ be a matrix and $\mt{M}_1$ ($\mt{M}_2$, resp.) be another matrix of the same height (width, resp.) as $\mt{M}$, $[\mt{M}, \mt{M}_1]$ ($[\mt{M}; \mt{M}_2]$, resp.) denotes the matrix obtained by concatenating $\mt{M}$ and $\mt{M}_1$ horizontally (concatenating  $\mt{M}$ and $\mt{M}_2$ vertically, resp.). Further, $\mm{N}(\mt{M})$ is the null space of $\mt{M}$ and $\vert \mt{M} \vert$ is the matrix that consists of the element-wise absolute values of $\mt{M}$.

Let $\underline{\mt{a}}$, $\overline{\mt{a}} \in \mathbb{R}^n$ such that $\underline{\mt{a}}\leq \overline{\mt{a}}$ ($\leq $ is element-wise), a \emph{hyper-box} $\li \underline{\mt{a}}, \overline{\mt{a}}\ri$ is the set $\{\mt{x} \in \mathbb{R}^n\mid \underline{\mt{a}} \leq \mt{x} \leq \overline{\mt{a}}\}$.  
Let $\mt{G} \in \mathbb{R}^{n \times N}$ and $\mt{c} \in \mathbb{R}^{n}$, a \emph{zonotope} $\mm{Z} = \langle \mt{G}, \mt{c}\rangle $ is defined to be the set
$\big\{\mt{G}\mtg{\theta} + \mt{c} \mid \mtg{\theta} \in \li-\bfo, \bfo\ri\}$. 
The tuple $\langle \mt{G}, \mt{c}\rangle$ is called the \emph{generator-representation} (or G-Rep) of $\mm{Z}$. 
The matrix $\mt{G}$ is the \emph{generator matrix} and $\mt{c}$ is the \emph{center} of $\mm{Z}$. 
A set $\mm{CZ}$ is a \emph{constrained zonotope} 
if it can be expressed as $\{\mt{G}\mtg{\theta} + \mt{c} \mid  \mtg{\theta}  \in \li-\bfo,\bfo\ri, \mt{A}\mtg{\theta} = \mt{b}\}$, where $\mt{A}\in \mathbb{R}^{m \times N}$ and $\mt{b} \in \mathbb{R}^m$. The tuple $\langle \mt{G},\mt{c},\mt{A},\mt{b}\rangle$ is a \emph{constrained generator representation} (or CG-Rep) of $\mm{CZ}$, $\mt{A}$ is the \emph{constraint matrix} and $\mt{b}$ is the \emph{constraint vector} of this CG-Rep. 
A zonotope $\langle \mt{G}, \mt{c}\rangle$ is a constrained zonotope whose CG-Rep has the same $\mt{G}$, $\mt{c}$ and empty $\mt{A}$, $\mt{b}$. 
Further, let $\mt{H} \in \mathbb{R}^{\ell\times N}$ and $\mt{a} \in \mathbb{R}^\ell$,  
a set is an \emph{AH-polytope} if it can be expressed as $\{\mt{G}\mtg{\theta} + \mt{c} \mid \mt{H} \mtg{\theta} \leq \mt{a}\}$. Zonotopes and constrained zonotopes are  AH-polytopes, i.e., $\langle \mt{G}, \mt{c}\rangle = \{\mt{G}\mtg{\theta} + \mt{c} \mid [\mt{I};-\mt{I}] \mtg{\theta} \leq \bfo\}$ and $\langle \mt{G},\mt{c},\mt{A},\mt{b}\rangle \cgeq \{\mt{G}\mtg{\theta} + \mt{c} \mid [\mt{A};-\mt{A};\mt{I};-\mt{I}] \mtg{\theta} \leq [\mt{b};-\mt{b};\bfo]\}$. 


Let $\mm{S}$, $\mm{R} \subseteq \mathbb{R}^n$ be two sets, $\mt{x} \in \mathbb{R}^n$ be a vector and $\mt{M} \in \mathbb{R}^{m \times n}$ be a matrix, we define $\mt{M} \,\mm{S} : = \{\mt{M}\mt{s} \mid \mt{s} \in \mm{S}\}$ and $\mt{x} + \mm{S} := \{ \mt{x} + \mt{s} \mid \mt{s} \in \mm{S}\}$. Further, $\mm{S} \oplus \mm{R} := \{\mt{s} + \mt{r} \mid \mt{s} \in \mm{S}, \mt{r} \in \mm{R}\}$ is the \emph{Minkowski sum} of $\mm{S}$ and $\mm{R}$, and $\mm{S} \ominus \mm{R} := \{\mt{x} \in \mathbb{R}^n \mid \mt{x} + \mm{R} \subseteq \mm{S}\}$ is the \emph{Minkowski difference}  between $\mm{S}$ and $\mm{R}$. Let $\mm{P}\subseteq \mathbb{R}^p$, $\mm{S}\times \mm{P} := \{[\mt{s};\mt{p}]\mid \mt{s} \in \mm{S}, \mt{p} \in \mm{P}\}$ is the \emph{product} of $\mm{S}$ and $\mm{P}$. 

The following 
set operations 
can be performed by CG-Rep manipulation for constrained zonotopes. 

\begin{lem} \label{lem:cz}
[From \cite{raghuraman2020set,scott2016constrained}] 
Let $\mm{CZ} \cgeq \langle \mt{G},\mt{c},\mt{A},\mt{b}\rangle \subseteq \mathbb{R}^n$, $\mm{CZ}_i \cgeq \langle \mt{G}_i,\mt{c}_i,\mt{A}_i,\mt{b}_i\rangle  \subseteq \mathbb{R}^p$ for $i \in \{1,2\}$ 
be  constrained zonotopes, $\mt{M} \in \mathbb{R}^{m \times n}$ be a matrix and $\mm{H} = \{\mt{x} \in \mathbb{R}^{n} \mid \mt{h}^\top \mt{x} \leq a\}$ be a halfspace,  then
\begin{itemize}
\item[i)] $\mt{M}\,\mm{CZ} \cgeq \langle \mt{M}\mt{G}, \mt{M}\mt{c}, \mt{A}, \mt{b}\rangle$, 
\item[ii)] $\mm{CZ}_1 \oplus \mm{CZ}_2 \cgeq \langle[\mt{G}_1,\mt{G}_2], \mt{c}_1+\mt{c}_2, \text{diag}(\mt{A}_1,\mt{A}_2), $ $[\mt{b}_1;\mt{b}_2]\rangle$, 
\item[iii)] $\mm{CZ}_1 \cap \mm{CZ}_2 \cgeq \langle[\mt{G}_1,\bfz], \mt{c}_1, [\text{diag}(\mt{A}_1,\mt{A}_2); [\mt{G}_1, -\mt{G}_2]],$ $[\mt{b}_1;\mt{b}_2;\mt{c}_2-\mt{c}_1]\rangle$, 
\item[iv)]  if $\mm{CZ} \cap \mm{H} \neq \emptyset$, then $\mm{CZ} \cap \mm{H} \cgeq \langle [\mt{G}, \bfz], \mt{c}, $ $[\mt{A}, \bfz; \mt{h}^\top \mt{G},  \tfrac{d}{2}], $ $[\mt{b}; a -\mt{h}^\top \mt{c} - \tfrac{d}{2}]\rangle $ where $d = a - \mt{h}^\top \mt{c} + \Vert \mt{h}^\top \mt{G}\Vert_1$,
\item[v)]  $\mm{CZ} \times \mm{CZ}_1 \cgeq \langle \text{diag}(\mt{G}, \mt{G}_1), [\mt{c};\mt{c}_1],\text{diag}(\mt{A}, \mt{A}_1),[\mt{b};\mt{b}_1] \rangle$. 
\end{itemize}
\end{lem}

The following lemma follows from the definitions. 
\begin{lem}\label{lem:Minkowski}
Let $\mm{A}$, $\mm{B}$, $\mm{C}\subseteq \mathbb{R}^{N}$ and $\mt{M} \in \mathbb{R}^{n \times N}$
\begin{itemize}
\item[i)] $\mt{M}\,(\mm{A} \oplus \mm{B}) = \mt{M}\,\mm{A} \oplus \mt{M}\, \mm{B}$.
\item[ii)] $\mt{M}\,(\mm{A} \ominus \mm{B}) \subseteq \mt{M}\,\mm{A} \ominus \mt{M}\, \mm{B}$.
\item[iii)] $\left( \mm{A}\ominus \mm{C} \right) \cup \left( \mm{B} \ominus \mm{C} \right) \subseteq \left( \mm{A} \cup \mm{B} \right) \ominus \mm{C}$. 
\item[iv)] $\mm{A} \ominus \mm{B} \oplus \mm{C} \subseteq \mm{A} \oplus \mm{C} \ominus \mm{B}$.  
\end{itemize}
For ii), iii) and iv), 
equality does not hold in general. 
\end{lem}

\begin{lem} \label{lem:ahpc}
[\cite{sadraddini2019linear}, Theorem 1] 
Let $\mm{S}_i:= \mt{c}_i + \mt{G}_i\{\mtg{\theta}_i\mid \mt{H}_i \mtg{\theta}_i \leq \mt{a}_i\}\subseteq \mathbb{R}^n$ for $i \in \{1,2\} $ be two AH-polytopes. Suppose that $\mm{S}_1$ has nonempty interior. Then a sufficient condition for $\mm{S}_1 \subseteq \mm{S}_2$ is that there exist matrices $\mtg{\Gamma}, \mtg{\beta}, \mtg{\Lambda}$ of proper sizes such that 
\begin{align}
&  \mt{G}_1 = \mt{G}_2 \mtg{\Gamma}, \ \mt{c}_2 - \mt{c}_1 = \mt{G}_2 \mtg{\beta}, \ \mtg{\Lambda} \mt{H}_1 = \mt{H}_2 \mtg{\Gamma}, \nonumber \\
&  \mtg{\Lambda} \mt{a}_1 \leq \mt{a}_2 + \mt{H}_2 \mtg{\beta}, \ \mtg{\Lambda} \geq \bfz. 
\label{eq:ahpc}
\end{align} 
\end{lem}
The condition in Eq. \eqref{eq:ahpc} is known as the encoding of AH-polytope containment. 
The numbers of variables and constraints in Eq. \eqref{eq:ahpc} are polynomial in the sizes of $\mt{c}_i$, $\mt{G}_i$, $\mt{H}_i$, $\mt{a}_i$. 
Since the zonotope containment problem, which is a special instance of the AH-polytope containment problem, is known to be co-NP hard \cite{kulmburg2021co}, the linear condition in Eq. \eqref{eq:ahpc} cannot possibly be necessary in general unless $\text{P}=\text{NP}$. 

 \begin{lem} \label{lem:HmB}
[\cite{kolmanovsky1998theory}, Theorem 2.3]
Let $\mm{S}\subseteq \mathbb{R}^n$ be compact, then 
$\{\mt{x} \mid \mt{H} \mt{x} \leq \mt{a}\} \ominus \mm{S} = \{x \mid \mt{H} \mt{x} \leq \mt{a} - \underline{\mt{a}}\}$, where $\underline{a}_i = \max \mt{h}_i^\top \mm{S}$ is the $i^{\rm th}$ element of $\underline{\mt{a}}$ and $\mt{h}_i^\top$ is the $i^{\rm th}$ row of $\mt{H}$.  
\end{lem}

\section{Problem Description}\label{sec:prob} 
Consider the following discrete-time nonlinear system: 
\begin{align}
\mt{x}_{t+1} = \mt{f}(\mt{x}_t, \mt{u}_t) + \mt{w}_t, \label{eq:sysnl} 
\end{align}
where $\mt{x} \in \mathbb{R}^n$ is the state, $\mt{u} \in \mm{U} \subseteq \mathbb{R}^q$ is the control input, and $\mt{w} \in \mm{W} \subseteq \mathbb{R}^n$ is the additive disturbance input. 
Given a set $\mm{X}_{\rm safe}$ of safe states and a set $\mm{X}_0$ of target states, the $k^{\rm th}$ backward reachable set $\mm{X}_k$ is defined recursively as follows: 
\begin{align}
\mm{X}_{k} & = Pre(\mm{X}_{k-1})
  := \left\{\begin{array}{l}\mt{x} \in \\ \mm{X}_{\rm safe} \end{array} \Big\vert
  \begin{array}{l}
  \exists \mt{u} \in \mm{U}: \forall \mt{w} \in \mm{W}: \\
  \mt{f}(\mt{x}, \mt{u}) + \mt{w} \in \mm{X}_{k-1}
  \end{array} \right\}, \nonumber \\
  k & =1,2,3\dots
\end{align}

Our goal is to compute $\underline{\mm{X}}_k$, represented by constrained zonotopes, s.t.   $\underline{\mm{X}}_k\subseteq \mm{X}_k$ under the following assumptions. 
\begin{assump}\label{assump:set}
The sets $\mm{X}_0$, $\mm{U}$ are  constrained zonotopes (CG-Reps given). The disturbance set $\mm{W}$ is a zonotope (G-Rep given). The safe set $\mm{X}_{\rm safe} = \bigcup_{p}\{\mt{x} \mid \mt{H}_p \mt{x} \leq \mt{a}_p\}$ is the union of finitely many polytopes in their H-Reps.
\end{assump}

Our solution approach uses sequential linearization. 
If the system is linear, i.e.,  $\mt{f}(\mt{x}, \mt{u}) = \mt{A}\mt{x} + \mt{B}\mt{u}$ for some invertible matrix $\mt{A}$ \footnote{The matrix $\mt{A}$ here is not to be confused with the constraint matrix in the CG-Rep of a constrained zonotope. }, then 
\begin{align}
    \mm{X}_{k} = \mm{X}_{\rm safe} \cap \mt{A}^{-1}(\mm{X}_{k-1}\ominus \mm{W} \oplus -\mt{B}\mm{U}). 
    \label{eq:brs_lin}
\end{align}
For nonlinear systems, in each step,  we linearize $\mt{f}$ at some $[\mt{x}^\ast; \mt{u}^\ast]$ and compute a under-approximation $\underline{\mm{X}}_k$ of $\mm{X}_k$ by applying Eq. \eqref{eq:brs_lin} to the previously obtained $\underline{\mm{X}}_{k-1}$ and the linear dynamics. 
Particularly, to ensure that $\underline{\mm{X}}_{k-1}$ can be reached from  $\underline{\mm{X}}_k$ under the nonlinear dynamics, 
we conservatively approximate the linearization error by an additive term that takes value from a zonotopic set $\mm{L}$, 
and require $\underline{\mm{X}}_{k-1} \ominus  \mm{L}$ to be reachable from $\underline{\mm{X}}_k$ under the linear dynamics (i.e., replace $\mm{W}$ in Eq. \eqref{eq:brs_lin} by $\mm{L}\oplus \mm{W}$).

The main challenge in the above approach is twofold:
\begin{itemize}
    \item[C1)] To implement Eq. \eqref{eq:brs_lin} under Assumption \ref{assump:set}, while 
    the affine transformation, intersection and Minkowski addition can be done via CG-Rep manipulation (Lemma \ref{lem:cz}), there still lacks an efficient way to compute (or to under-approximate) the Minkowski difference between a constrained zonotopic minuend  $\mm{X}_{k-1}$ and a zonotopic subtrahend $\mm{W}$ (or $\mm{L}\oplus \mm{W}$). \revise{Subsequent computations require this  Minkowski difference to be in CG-Rep. } 
    \item[C2)] To compute $\underline{\mm{X}}_k$, one needs to make a guess of $\mm{L}$ that encompasses all possible values of the additive linearization error over $\underline{\mm{X}}_k$, without knowing $\underline{\mm{X}}_k$ a priori. 
\end{itemize}
The rest of the paper is devoted to tackling these two challenges. 
In Sec. \ref{sec:Mdiff}, we develop an efficient algorithm for Minkowski difference under-approximation. 
We further show that our algorithm is exact for the problem instances whose minuend set has sufficiently rich CG-Reps (Sec. \ref{sec:rep}). In Sec. \ref{sec:brs_nl}, we explore two strategies to tackle challenge C2) and present two detailed algorithms that combine all ingredients together for BRS under-approximation.

\section{Under-approximating $\mCZ\ominus \sZ$}\label{sec:Mdiff} 
This section is concerned with under-approximating $\mCZ \ominus \sZ$, where $\mCZ \cgeq \langle \mt{G},\mt{c},\mt{A},\mt{b}\rangle \subseteq \mathbb{R}^{n}$ is a constrained zonotope and $\sZ = \langle \mt{G}',\mt{c}'\rangle \subseteq \mathbb{R}^{n}$ is a zonotope. 
\revise{We will show that computing a compact CG-rep of the \emph{exact} Minkowski difference is hard (Sec. \ref{sec:hard}). }
Hence we restrict our under-approximation to be a constrained zonotope $\mm{CZ}_{\rm d}$ that shares the same ``template'' as the minuend  $\mCZ$, i.e., $\mm{CZ}_{\rm d} \cgeq \langle \mt{G}\, \text{diag}(\overline{\mtg{\delta}}), \mt{c}_{\rm d}, \mt{A}\,\text{diag}(\overline{\mtg{\delta}}), \mt{b}_{\rm d} \rangle$ for some $\overline{\mtg{\delta}} \in \li \bfz, \bfo\ri$, $\mt{c}_{\rm d}\in \mathbb{R}^n$ and $\mt{b}_{\rm d}\in \mathbb{R}^m$. 
Under such restrictions, one can enforce $\mm{CZ}_{\rm d} \oplus \sZ \subseteq \mCZ$ using the constraints given by Lemma \ref{lem:ahpc} and find $\mm{CZ}_{\rm d}$ by solving an optimization problem. 
However, this optimization problem, as will be shown in Sec. \ref{sec:naive}, is a bilinear program. 
We refer to the above approach as the ``na\"ive approach''. 
To find a $\mm{CZ}_{\rm d}$ more efficiently, in Sec.  \ref{sec:twostep}, we propose a two-step approach that amounts to solving a linear program. 
We further show how to reduce the size of this linear program and present some results to understand how conservative our under-approximation is.
 
 \revise{
\subsection{Complexity Analysis}\label{sec:hard}
We show that, given the CG-Reps of $\mm{CZ}$ and $\mm{Z}$, it is impossible to find a polynomial-size CG-Rep of $\mm{CZ}\ominus \mm{Z}$ in polynomial time, unless P $=$ NP. 
This motivates us to find an under-approximation that admits a polynomial-size CG-Rep computable in polynomial time.  

\begin{prop}\label{prop:hard}
\normalfont
No algorithm satisfies the following two conditions simultaneously  unless P $=$ NP.  
\begin{itemize}[nolistsep]
    \item[a)] It finds $\langle \mt{G}'', \mt{c}'', \mt{A}'',\mt{b}''\rangle = \mm{CZ}\ominus \mm{Z}$ in ${\tt{poly}}(n, N, N')$  time, where $N$ ($N'$, resp.) is the width of $\mt{G}$ ($\mt{G}'$, resp.). 
    \item[b)] The widths and heights of matrices  $\mt{G}'', \mt{c}'', \mt{A}'',\mt{b}''$ are ${\tt{poly}}(n, N, N')$. 
\end{itemize} 
\end{prop}

\begin{proof}
Assume that an algorithm {\tt{A}} satisfies conditions a) and b) simultaneously. Since $\mm{Z} \subseteq \mm{CZ}$ iff $\mt{0} \in \mm{CZ} \ominus \mm{Z}$, whether $\mm{Z} \subseteq \mm{CZ}$ can be determined via the following procedure: 
\begin{itemize}[nolistsep]
    \item[1)] find $\langle \mt{G}'', \mt{c}'', \mt{A}'',\mt{b}''\rangle$ by algorithm {\tt{A}},  
    \item[2)] claim $\mm{Z} \subseteq \mm{CZ}$ iff $\mt{0} \in \langle \mt{G}'', \mt{c}'', \mt{A}'',\mt{b}''\rangle$.  
\end{itemize}
By bullet a), step 1) takes ${\tt{poly}}(n, N, N')$ time to run. 
Further, step 2) amounts to solving the following linear program: 
\begin{align}
    \begin{array}{rl}
        \text{find} & \mtg{\theta} \\
         \text{s.t.}&  \mt{G}''\mtg{\theta} + \mt{c}'' = \mt{0}, \ \mt{A}''\mtg{\theta} =  \mt{b}'',  \  -\mt{1} \leq \mtg{\theta} \leq \mt{1}
    \end{array} 
    \tag{LP}
    \label{eq:lp}
\end{align}
Let $N''$ ($m''$, resp.) be the width (height, resp.) of $\mt{A''}$,  there are $N''$ variables and $2N'' + m'' + n$ constraints in \eqref{eq:lp}. These two numbers are ${\tt{poly}}(n, N, N')$  by bullet b).  
Therefore, the above two-step procedure takes ${\tt{poly}}(n, N, N')$ time to run.
However, it is co-NP hard \cite{kulmburg2021co} to decide if $\mm{Z}\subseteq \mm{CZ}$ given the CG-Reps of $\mm{Z}$ and $\mm{CZ}$ as the inputs, which consist of $n (N + N' + 2)$ reals. Hence the existence of such an algorithm ${\tt{A}}$ that satisfies a) and b) implies P $=$ NP. 
\end{proof}
}

 \subsection{Na\"ive Approach with Bilinear Constraints}\label{sec:naive}

With the aforementioned na\"ive approach, we need to solve the following optimization problem: 
\begin{align}
\begin{array}{rl}
\max_{\overline{\mtg{\delta}}, \mt{c}_{\rm d}, \mt{b}_{\rm d}} &  \Vert\overline{\mtg{\delta}}\Vert_1  \\
\text{s.t. }&  \mm{CZ}_{\rm d} \oplus \sZ \subseteq \mCZ
\end{array}. 
\label{eq:blnr}
\end{align}
The objective function $\Vert \overline{\mtg{\delta}} \Vert_1$ is used as a heuristic to maximize the set $\mm{CZ}_{\rm d}$. 
To apply Lemma \ref{lem:ahpc}, we write $\mCZ$ as an AH-polytope, i.e., 
\begin{align}
\mCZ & = \mt{c} + \mt{G}\{\mtg{\theta} \mid [\mt{A};-\mt{A};\mt{I};-\mt{I}] \mtg{\theta} \leq [\mt{b};-\mt{b};\bfo]\}, \label{eq:e1}
\end{align}
and write $\mm{CZ}_{\rm d} \oplus \mm{Z}$ either as 
\begin{align}
&  \ \ \ \ \ \mt{c}_{\rm d} \hspace{-0mm} +  \hspace{-0mm}\mt{c}' \hspace{-0mm} +  \hspace{-0mm}[\mt{G}\,\text{diag}(\overline{\mtg{\delta}}), \mt{G}'] \{\mtg{\xi} \hspace{-0mm}\mid 
\hspace{-0mm} [\mt{A}\,\text{diag}(\overline{\mtg{\delta}}),\bfz; \nonumber\\ 
& \ \ \ \ \ \ \ \ \ \ -\mt{A}\,\text{diag}(\overline{\mtg{\delta}}), \bfz; \mt{I};-\mt{I}] \mtg{\xi}
 \leq [\mt{b}_{\rm d};-\mt{b}_{\rm d};\bfo]\},  \label{eq:e2} 
\end{align}
 or as 
\begin{align}
&  \ \ \ \ \ \mt{c}_{\rm d} \hspace{-0mm} +  \hspace{-0mm}\mt{c}' \hspace{-0mm} +  \hspace{-0mm}[\mt{G}, \mt{G}']  \{\mtg{\xi} \mid [\mt{A},\bfz; -\mt{A},\bfz;\mt{I};-\mt{I}] \mtg{\xi} \nonumber \\
& \ \ \ \ \ \ \ \ \ \ \leq [\mt{b}_{\rm d};-\mt{b}_{\rm d};\overline{\mtg{\delta}};\bfo;\overline{\mtg{\delta}};\bfo]\}.  \label{eq:e3}
\end{align}
Unfortunately, Lemma \ref{lem:ahpc} gives bilinear constraints when applied to Eqs. \eqref{eq:e1},\eqref{eq:e2} or to Eqs. \eqref{eq:e1},\eqref{eq:e3}. 
If \eqref{eq:e2} is used, ``$\mt{H}_1$'' in \eqref{eq:ahpc} depends on the variable $\overline{\mtg{\delta}}$ and the term ``$\mtg{\Lambda} \mt{H}_1$" is bilinear; if \eqref{eq:e3} is used, ``$\mt{a}_1$'' in \eqref{eq:ahpc} depends on $\overline{\mtg{\delta}}$ and ``$\mtg{\Lambda} \mt{a}_1$'' is bilinear.

The key observation here is that the encoding in Lemma \ref{lem:ahpc} is more favorable (i.e., tends to be linear) if the variables are related to the \emph{outer} set. 
On the contrary, the above encoding is bilinear because the variable $\overline{\mtg{\delta}}$ is related to the \emph{inner} set.

 \subsection{Two-Step Approach: Overview}\label{sec:twostep}

We propose an alternative approach that finds an under-approximation $\mm{CZ}_{\rm d}$ of  $\mCZ \ominus \sZ$ with the following two steps. 
\begin{itemize} 
\item[I)] Compute a vector $\overline{\mtg{\sigma}} \in \li \bfz, \bfo\ri$ such that $\mm{CZ}_{\rm s} \cgeq \langle \mt{G}\,\text{diag}(\overline{\mtg{\sigma}}), \mt{c}_{\rm s}, A\,\text{diag}(\overline{\mtg{\sigma}}), \mt{b}_{\rm s}\rangle$ encloses 
$\mm{Z}$. 
\item[II)] Compute $\mm{CZ}_{\rm d} \cgeq \langle \mt{G}\,\text{diag}(\bfo - \overline{\mtg{\sigma}}), \mt{c}-\mt{c}_{\rm s}, \mt{A}\,\text{diag}(\bfo - \overline{\mtg{\sigma}}), \mt{b}-\mt{b}_{\rm s}\rangle$. 
\end{itemize}
Since $\sZ \subseteq \mm{CZ}_{\rm s}$ by construction, $\mCZ \ominus \mm{CZ}_{\rm s} $ is an under-approximation of $\mCZ \ominus \sZ$. 
The significance of Step I) is that, since the variable $\overline{\mtg{\sigma}}$ is related to the \emph{outer} set $\mm{CZ}_{\rm s}$, the encoding of $\sZ \subseteq \mm{CZ}_{\rm s}$ by Lemma \ref{lem:ahpc} is linear. 
Further, since the generator matrix $\mt{G}$ and the constraint matrix $\mt{A}$ of the minuend  $\mCZ$ are used as ``templates'' when constructing $\mm{CZ}_{\rm s}$, it follows that $\mCZ \ominus \mm{CZ}_{\rm s}  \supseteq \mm{CZ}_{\rm d} \cgeq \langle \mt{G}\,\text{diag}(\bfo - \overline{\mtg{\sigma}}), \mt{c}-\mt{c}_{\rm s}, \mt{A}\,\text{diag}(\bfo - \overline{\mtg{\sigma}}), \mt{b}-\mt{b}_{\rm s}\rangle$. 
Hence $\mCZ \ominus \mm{CZ}_{\rm s}$  can be further under-approximated via a simple CG-Rep manipulation. 
The above two-step approach extends the one in \cite{yang2022scalable} for under-approximating the Minkowski difference of two zonotopes.

In what follows, we first show in details how to implement Step I) by solving a linear program. 
We further simplify this linear program by showing that it is optimal to choose $\mt{c}_{\rm s} = \mt{c}'$ and $\mt{b}_{\rm s} = \bfz$ in Step I). 
Then we prove that $\mCZ \ominus \mm{CZ}_{\rm s} \supseteq \mm{CZ}_{\rm d}$.  
As a step to understand the conservatism of our under-approximation, we will also give a sufficient condition for $\mCZ \ominus \mm{CZ}_{\rm s} = \mm{CZ}_{\rm d}$  
to hold.

\subsection{Step I: Over-approximating $\mm{Z}$ by $\mm{CZ}_{\rm s}$}
Our goal is to solve 
\begin{align}
\begin{array}{rl}
\min_{\overline{\mtg{\sigma}}, \mt{c}_{\rm s}, \mt{b}_{\rm s}} &  \Vert\overline{\mtg{\sigma}}\Vert_1  \\
\text{s.t. }& \sZ \subseteq \mm{CZ}_{\rm s}
\end{array}, 
\label{eq:minout0}
\end{align}
where  $\sZ = \langle \mt{G}', \mt{c}'\rangle$ and $\mm{CZ}_{\rm s} \cgeq \langle \mt{G}\,\text{diag}(\overline{\mtg{\sigma}}), \mt{c}_{\rm s}, $ $\mt{A}\,\text{diag}(\overline{\mtg{\sigma}}), \mt{b}_{\rm s}\rangle$. 
In \eqref{eq:minout0}, we minimize $\Vert\overline{\mtg{\sigma}}\Vert_1$. This can be seen as a heuristic to minimize the enclosing constrained zonotope $\mm{CZ}_{\rm s}$. 
Note that $\mm{CZ}_{\rm s}$ can be rewritten as $\mt{c}_{\rm s} + \mt{G} \mm{S}_{\mt{b}_{\rm s}}$ where $
\mm{S}_{\mt{b}_{\rm s}} =  \{\mtg{\sigma} \in \li -\overline{\mtg{\sigma}}, \overline{\mtg{\sigma}} \ri \mid \mt{A} \mtg{\sigma} = \mt{b}_{\rm s}\}$.

The following result shows that, to solve the optimization problem in \eqref{eq:minout0}, 
one can choose $\mt{c}_{\rm s} = \mt{c}'$ and $\mt{b}_{\rm s} = \bfz$ without loss of optimality. 

\begin{prop}\label{prop:cb}
Let $\mm{S}_{\mt{b}_{\rm s}} : = \{\mtg{\sigma} \in \li -\overline{\mtg{\sigma}}, \overline{\mtg{\sigma}} \ri \mid \mt{A} \mtg{\sigma} = \mt{b}_{\rm s}\}$ and $\mm{Z}$ be a zonotope centering at $\mt{c}'$. 
We have 
\begin{align}
& \ \min \{\Vert\overline{\mtg{\sigma}}\Vert_1 \mid \sZ \subseteq \mt{c}' + \mt{G} \mm{S}_{\bfz} \} \nonumber \\
\leq & \ \min \{\Vert\overline{\mtg{\sigma}}\Vert_1 \mid \exists \mt{c}_{\rm s}, \mt{b}_{\rm s} : \sZ \subseteq \mt{c}_{\rm s} + \mt{G} \mm{S}_{\mt{b}_{\rm s}}\}. 
\label{eq:mincb}
\end{align}
\end{prop}

\begin{proof}
We first prove that $\sZ \subseteq \mt{c}_{\rm s} + \mt{G}\,\mm{S}_{\mt{b}_{\rm s}}$ implies $\sZ \subseteq \mt{c}' + \mt{G} \mm{S}_{\bfz}$. 
Note that $\sZ$ is symmetric w.r.t its center $\mt{c}'$, i.e., $-\sZ+\mt{c}' = \sZ-\mt{c}'$. 
By $\sZ \subseteq \mt{c}_{\rm s} + \mt{G}\,\mm{S}_{\mt{b}_{\rm s}}$, we have 
\begin{align}
& \ \sZ - \mt{c}'  \nonumber \\
\subseteq & \ (\mt{c}_{\rm s}-\mt{c}')+ \mt{G}\,\mm{S}_{\mt{b}_{\rm s}} \cap -(\mt{c}_{\rm s}-\mt{c}')- \mt{G}\,\mm{S}_{\mt{b}_{\rm s}}\nonumber \\
=& \left\{
\begin{array}{c}
\mt{G}\mtg{\theta} + \\
(\mt{c}_{\rm s}-\mt{c}') 
\end{array}
  \bigg\vert
\begin{array}{l}
\mtg{\theta}, \mtg{\mu} \in \li-\overline{\mtg{\sigma}}, \overline{\mtg{\sigma}}\ri,  \mt{A}\mtg{\theta} = \mt{A}\mtg{\mu} = \mt{b}_{\rm s}, \\
 \mt{G}\mtg{\theta} + (\mt{c}_{\rm s}-\mt{c}') = -\mt{G}\mtg{\mu} - (\mt{c}_{\rm s}-\mt{c}')
\end{array}
 \right\}   \nonumber \\
\subseteq &\, \Big\{\mt{G}\mtg{\theta} - \mt{G} \Big(\tfrac{\mtg{\mu} + \mtg{\theta}}{2}\Big) \, \Big\vert \mtg{\theta}, \mtg{\mu} \in \li-\overline{\mtg{\sigma}}, \overline{\mtg{\sigma}}\ri,  \mt{A}(\mtg{\theta}-\mtg{\mu}) = \bfz\Big\} \nonumber \\
=  &\, \Big\{\mt{G} \Big(\tfrac{ \mtg{\theta}-\mtg{\mu}}{2}\Big) \, \Big\vert  \tfrac{ \mtg{\theta}-\mtg{\mu}}{2}  \in \li-\overline{\mtg{\sigma}}, \overline{\mtg{\sigma}}\ri,  \mt{A}\Big(\tfrac{ \mtg{\theta}-\mtg{\mu}}{2}\Big) = \bfz\Big\} \nonumber \\
= & \,\{\mt{G}\mtg{\sigma} \mid \mtg{\sigma} \in \li-\overline{\mtg{\sigma}}, \overline{\mtg{\sigma}}\ri, \mt{A}\mtg{\sigma} = \bfz\}
= \  \mt{G} \mm{S}_{\bfz}
\end{align} 
Therefore, we have $\sZ \subseteq \mt{c}' + G \mm{S}_{\bfz}$. Since $\mm{S}_{\mt{b}_{\rm s}}$ and $\mm{S}_{\bfz}$ are defined by the same $\overline{\mtg{\sigma}}$, Eq. \eqref{eq:mincb} follows readily. 
\end{proof}

 \begin{rmk}
Proposition \ref{prop:cb} holds for any set $\sZ$ that is symmetric w.r.t. $\mt{c}'$ and any other cost  function 
of $\overline{\mtg{\sigma}}$ than $\Vert \mtg{\overline{\sigma}} \Vert_1$. 
\end{rmk}

Next, we show that, with the sufficient condition for $\sZ \subseteq \mm{CZ}_{\rm s}$ given in Lemma \ref{lem:ahpc}, how to find a suboptimal solution of \eqref{eq:minout0} by solving a linear program.

\begin{prop}\label{prop:minout}
Suppose that  $\sZ = \langle \mt{G}',\mt{c}'\rangle$ has nonempty interior.  
Let $\overline{\mtg{\sigma}}$ be part of a minimizer of the following linear program: 
\begin{align}
\begin{array}{rl}
\hspace{-5mm}
 \min\limits_{\overline{\mtg{\sigma}}, \mt{c}_{\rm s}, \mt{b}_{\rm s}, \mtg{\Gamma}, \mtg{\beta}, \mtg{\Lambda}} &  \Vert\overline{\mtg{\sigma}}\Vert_1 \\ 
\text{ s.t. } &  \mt{G}' = \mt{G}\mtg{\Gamma}, \mt{G}\mtg{\beta} = \mt{c}_{\rm s} - \mt{c}', \\
& \mtg{\Lambda} [\mt{I};-\mt{I}] = [\mt{A};-\mt{A};\mt{I};-\mt{I}]\mtg{\Gamma}, \\
& \mtg{\Lambda} \bfo \leq [\mt{b}_{\rm s}; -\mt{b}_{\rm s}; \overline{\mtg{\sigma}}; \overline{\mtg{\sigma}}] +  [\mt{A};-\mt{A};\mt{I};-\mt{I}]\mtg{\beta},  \\
& \bfz \leq \overline{\mtg{\sigma}} \leq \bfo, \  \mtg{\Lambda}\geq \bfz
\end{array}
\tag{min-out}
\label{eq:minout}
\end{align}
then $\mm{Z} \subseteq \mm{CZ}_{\rm s} \cgeq \langle G\,\text{diag}(\overline{\mtg{\sigma}}), \mt{c}_{\rm s}, A\,\text{diag}(\overline{\mtg{\sigma}}), \mt{b}_{\rm s} \rangle$.  
\end{prop}
 
\begin{proof}
Note that the constrained zonotope $\mm{CZ}_{\rm s}$ and the zonotope $\mm{Z}$ can be written as 
\begin{align}
\hspace{-3mm}\mm{CZ}_{\rm s} 
& = \mt{c}_{\rm s} +\mt{G} \{\mtg{\sigma}   \mid   [\mt{A};-\mt{A};\mt{I};-\mt{I}]\mtg{\sigma} \leq [\mt{b}_{\rm s}; -\mt{b}_{\rm s}; \overline{\mtg{\sigma}}; \overline{\mtg{\sigma}}] \}, \label{eq:CZah}\\
\mm{Z} & = \mt{c}' + \mt{G}'\{\mtg{\sigma} \mid [\mt{I};-\mt{I}] \mtg{\sigma} \leq \bfo \}. \label{eq:Zah}
\end{align}
Therefore $\mm{CZ}_{\rm s} \supseteq \mm{Z}$ can be enforced by a set of linear constraints using Lemma \ref{lem:ahpc}, which leads to \eqref{eq:minout}.  
\end{proof}

In spite of Proposition \ref{prop:cb}, we keep $\mt{c}_{\rm s}$, $\mt{b}_{\rm s}$ as free variables in \eqref{eq:minout}. 
In what follows, we show that one can also set $\mt{c}_{\rm s} = \mt{c}'$ and $\mt{b}_{\rm s} = \bfz$ in \eqref{eq:minout} without loss of optimality. This result leads to a linear program equivalent to \eqref{eq:minout} with fewer variables and constraints. 
Note that this result does not follow immediately from Proposition \ref{prop:cb} because the condition in Lemma \ref{lem:ahpc} is only sufficient but not necessary in general (in fact, if that condition were also necessary, it would be straightforward that $\mt{c}_{\rm s} = \mt{c}'$ and $\mt{b}_{\rm s} = \bfz$ is optimal for \eqref{eq:minout}). 
The proof is based on the following observations. 

\begin{prop}\label{prop:obs}
Let $(\overline{\mtg{\sigma}}, \mt{c}_{\rm s}, \mt{b}_{\rm s}, \mtg{\Gamma}, \mtg{\beta}, \mtg{\Lambda})$ be a feasible solution of \eqref{eq:minout}, then i) $\mt{c}_{\rm s} = \mt{G} \mtg{\beta} - \mt{c}'$, ii) $\mt{b}_{\rm s} = - \mt{A} \mtg{\beta}$ and iii) $\mt{A} \mtg{\Gamma} = \bfz$, 
and iv) $(\overline{\mtg{\sigma}}, \mt{c}', \bfz, \mtg{\Gamma}, \bfz, \underline{\mtg{\Lambda}})$ is feasible for some  $\underline{\mtg{\Lambda}}$. 
\end{prop}

\begin{proof}
Bullet i) follows from the constraint $\mt{G}\mtg{\beta} = \mt{c}_{\rm s} - \mt{c}'$. 
By $\mtg{\Lambda} \geq \bfz$ (hence $\mtg{\Lambda} \bfo\geq \bfz$) and $\mtg{\Lambda} \bfo \leq [\mt{b}_{\rm s}; -\mt{b}_{\rm s}; \overline{\mtg{\sigma}}; \overline{\mtg{\sigma}}] +  [\mt{A};-\mt{A}; \mt{I}; -\mt{I}]\mtg{\beta}$, we have 
\begin{align}
\bfz\leq \mtg{\Lambda} \bfo \leq [\mt{b}_{\rm s}; -\mt{b}_{\rm s}; \overline{\mtg{\sigma}}; \overline{\mtg{\sigma}}] +  [\mt{A};-\mt{A}; \mt{I};-\mt{I}]\mtg{\beta}.
\end{align}
This implies that $\bfz \leq \mt{b}_{\rm s} + \mt{A}\mtg{\beta}$ and $\bfz \leq -\mt{b}_{\rm s}-\mt{A}\mtg{\beta}$, i.e., $\mt{b}_{\rm s} + \mt{A}\mtg{\beta} = \bfz$. 
Hence bullet ii) holds. 
Also, $\mtg{\Lambda} \bfo  = [\bfz_{2m}; \overline{\mtg{\sigma}} + \mtg{\beta}; \overline{\mtg{\sigma}} - \mtg{\beta}]$ and hence the upper part of matrix $\mtg{\Lambda}$ must be all zeros, i.e., $\mtg{\Lambda} = [\bfz_{2m\times 2N'};\widetilde{\mtg{\Lambda}}]$ for some $\widetilde{\mtg{\Lambda}} \geq \bfz$, where $m$ is the height of $\mt{A}$ and $N'$ is the width of $\mt{G}'$. 
This further implies that $\mt{A} \mtg{\Gamma}  = \bfz$ because $\mtg{\Lambda} [\mt{I};-\mt{I}] = [\mt{A};-\mt{A};\mt{I};-\mt{I}]\mtg{\Gamma}$. 

To prove bullet iv), define $\underline{\mtg{\Lambda}}$ as follows. The topmost $2m$ rows of $\underline{\mtg{\Lambda}}$ are all zeros. 
For $i = 1,2\dots N$,  where $N$ is the width of $\mt{G}$, 
\begin{itemize}
\item[i)] if the $i^{\rm th}$ element of $\mtg{\beta}$ is non-positive, we define the $2m+i^{\rm th}$ row of $\underline{\mtg{\Lambda}}$ to be the same as that of $\mtg{\Lambda}$, i.e., $[\mtg{\Lambda}_{2m+i,1:N'}, \mtg{\Lambda}_{2m+i,N'+1:2N'}]$, and the $2m+N+i^{\rm th}$ row of $\underline{\mtg{\Lambda}}$ to be $[\mtg{\Lambda}_{2m+i,N'+1:2N'}, \mtg{\Lambda}_{2m+i,1:N'}]$; 
\item[ii)]  if the $i^{\rm th}$ element of $\mtg{\beta}$ is positive, we define the $2m+N+i^{\rm th}$ row of $\underline{\mtg{\Lambda}}$ to be the same as that of $\mtg{\Lambda}$, i.e., $[\mtg{\Lambda}_{2m+N+i,1:N'}, \mtg{\Lambda}_{2m+N+i,N'+1:2N'}]$, and the $2m+i^{\rm th}$ row of $\underline{\mtg{\Lambda}}$ to be $[\mtg{\Lambda}_{2m+N+i,N'+1:2N'}, \mtg{\Lambda}_{2m+N+i,1:N'}]$. 
\end{itemize}
By construction, $\underline{\mtg{\Lambda}}$ has a special structure, i.e., $\underline{\mtg{\Lambda}} = [\bfz_{2m\times 2N'}; \mtg{\Lambda}_1, \mtg{\Lambda}_2; \mtg{\Lambda}_2, \mtg{\Lambda}_1]$. 
Moreover, $\underline{\mtg{\Lambda}} [\mt{I}; -\mt{I}] = [\bfz_{2m}; \mtg{\Gamma}; -\mtg{\Gamma}]$ and  $\underline{\mtg{\Lambda}} \bfo \leq [\bfz_{2m}; \overline{\mtg{\sigma}}-|\mtg{\beta}|; \overline{\mtg{\sigma}}-|\mtg{\beta}|] \leq [\bfz_{2m}; \overline{\mtg{\sigma}}; \overline{\mtg{\sigma}}]$. 
Together with bullet i) ii) and iii), it is straightforward to check that $(\overline{\mtg{\sigma}}, \mt{c}', \bfz, \mtg{\Gamma}, \bfz, \underline{\mtg{\Lambda}})$ is feasible. 
\end{proof}

Proposition \ref{prop:obs} leads to a simplification of \eqref{eq:minout}.

\begin{thm}\label{thm:simple}
The linear program \eqref{eq:minout} is equivalent to
\begin{align}
\begin{array}{rl}
\hspace{-5mm}\min_{\mtg{\Gamma}} &   \Vert\vert\mtg{\Gamma}\vert \bfo\Vert_1 \\
\text{ s.t. } &  [\mt{G};  \mt{A}]\mtg{\Gamma} = [\mt{G}'; \bfz], \ \vert\mtg{\Gamma}\vert \bfo \leq \bfo
\end{array}. 
\tag{simple}
\label{eq:minouts}
\end{align}
\end{thm}

\begin{proof} 
If $\mtg{\Gamma}$ minimizes \eqref{eq:minouts}, $(\overline{\mtg{\sigma}}, \mt{c}', \bfz, \mtg{\Gamma}, \bfz, \underline{\mtg{\Lambda}})$ is feasible to \eqref{eq:minout}, where $\overline{\mtg{\sigma}}=\vert\mtg{\Gamma}\vert \bfo$, $\underline{\mtg{\Lambda}} = [\bfz_{2m\times 2N'}; \mtg{\Lambda}_1, \mtg{\Lambda}_2; \mtg{\Lambda}_2, \mtg{\Lambda}_1]$, $\mtg{\Lambda}_1 = \mtg{\Gamma}^+ : = (\vert\mtg{\Gamma}\vert + \mtg{\Gamma})/2$ and $\mtg{\Lambda}_2 = \mtg{\Gamma}^- : = (\vert\mtg{\Gamma}\vert - \mtg{\Gamma})/2$. 
Moreover, the cost given by this feasible solution is $\Vert \overline{\mtg{\sigma}} \Vert_1 = \Vert\vert\mtg{\Gamma}\vert \bfo \Vert_1$, i.e., the same as the minimum of \eqref{eq:minouts}.  
Suppose that $(\overline{\mtg{\sigma}}, \mt{c}_{\rm s}, \mt{b}_{\rm s}, \mtg{\Gamma}, \mtg{\beta}, \mtg{\Lambda})$ minimizes  \eqref{eq:minout}.  
Construct $\underline{\mtg{\Lambda}}$ from $\mtg{\Lambda}$ as in the proof of Proposition \ref{prop:obs}, and let $[\mtg{\Lambda}_1,  \mtg{\Lambda}_2]$ consist of the $2m+1^{\rm st}$ to $2m+N^{\rm th}$ rows of $\underline{\mtg{\Lambda}}$. 
Then $\mtg{\Gamma} = \mtg{\Lambda}_1 - \mtg{\Lambda}_2$ is feasible to \eqref{eq:minouts}. 
Further, the cost associated with $\mtg{\Gamma}$ is  $\Vert\vert\mtg{\Gamma}\vert \bfo \Vert_1 \leq \Vert \overline{\mtg{\sigma}} \Vert_1$.
This is because, by Proposition \ref{prop:obs}, $(\overline{\mtg{\sigma}}, \mt{c}', \bfz, \mtg{\Gamma}, \bfz, \underline{\mtg{\Lambda}})$ is also feasible to \eqref{eq:minout} and hence $\overline{\mtg{\sigma}} \geq \mtg{\Lambda}_1 \bfo + \mtg{\Lambda}_2 \bfo = \vert \mtg{\Gamma} \vert \bfo$ must hold.
\end{proof}

By Theorem \ref{thm:simple}, we can find the minimizer $\mtg{\Gamma}$ of \eqref{eq:minouts}, define $\overline{\mtg{\sigma}} = \vert\mtg{\Gamma}\vert \bfo$ and $\langle \mt{G}\,\text{diag}(\overline{\mtg{\sigma}}),$ $ \mt{c}', \mt{A}\,\text{diag}(\overline{\mtg{\sigma}}), \bfz\rangle$ is guaranteed to enclose $\mm{Z}$. 
 From now no, we will use $\mm{CZ}_{\rm s}$ to denote the constrained zonotope $\langle \mt{G}\,\text{diag}(\overline{\mtg{\sigma}}), \mt{c}', \mt{A}\,\text{diag}(\overline{\mtg{\sigma}}), \bfz\rangle$ and omit the subscript ``$\bfz$'' of the set $\mm{S}_{\bfz} = \{\mtg{\sigma} \in  \li -\overline{\mtg{\sigma}}, \overline{\mtg{\sigma}}\ri\mid \mt{A} \mtg{\sigma} = \bfz\}$.

The simplified optimization problem \eqref{eq:minouts} has a geometric interpretation. 
Its decision variable $\mtg{\Gamma}$ can be viewed as the generator matrix of a zonotope $\langle \mtg{\Gamma}, \bfz\rangle \subseteq \mathbb{R}^{N}$, where $N$ is the width of $\mt{G}$. 
The inner zonotope $\mm{Z}$  is the image of $\langle \mtg{\Gamma}, \bfz\rangle$ under linear map $\mt{G}$ and translation $\mt{c}'$. Moreover, $\langle \mtg{\Gamma}, \bfz\rangle$ is in the null space of $\mt{A}$, and $\li - \overline{\mtg{\sigma}}, \overline{\mtg{\sigma}}\ri$ is the smallest hyper-box that encloses $\langle \mtg{\Gamma}, \bfz\rangle$ as $\overline{\mtg{\sigma}} = \vert \mtg{\Gamma} \vert \bfo$. With this interpretation, it is easy to see $\mm{Z} \subseteq \mm{CZ}_{\rm s}$. 
To be precise, 
\begin{align}
\mm{Z} = & \,  \langle \mt{G}', \mt{c}'\rangle \nonumber \\
= & \, \{\mt{G}\mtg{\Gamma}\mtg{\theta} + \mt{c}' \mid \mtg{\theta} \in \li -\bfo, \bfo \ri \} & (\mt{G}\mtg{\Gamma} = \mt{G}') \nonumber \\
= & \, \{\mt{G} \mtg{\sigma} + \mt{c}' \mid \mtg{\sigma} \in \langle \mtg{\Gamma}, \bfz\rangle\}\nonumber \\
= & \,  \{\mt{G} \mtg{\sigma} + \mt{c}' \mid \mtg{\sigma} \in \langle \mtg{\Gamma}, \bfz\rangle, \mt{A} \mtg{\sigma} = \bfz\} & (\mt{A}\mtg{\Gamma} = \bfz)\nonumber \\
\subseteq & \,  \{\mt{G} \mtg{\sigma} + \mt{c}' \mid \mtg{\sigma} \in \li - \overline{\mtg{\sigma}}, \overline{\mtg{\sigma}}\ri, \mt{A} \mtg{\sigma} = \bfz\} & (\overline{\mtg{\sigma}} = \vert\mtg{\Gamma}\vert \bfo) \nonumber \\
= & \, \mt{c}' + \mt{G} \mm{S} = \mm{CZ}_{\rm s}. 
\end{align}
 
\begin{rmk}\label{rmk:wlog}
For arbitrary constrained zonotope $\langle \mt{G}, \mt{c}, \mt{A}, \mt{b}\rangle$, the unit hyper-box $\li - \bfo, \bfo\ri$ may not necessarily be \emph{tight}, i.e., it is not the smallest hyper-box that  encloses $\{\mtg{\theta} \in \li - \bfo, \bfo\ri \mid \mt{A}\mtg{\theta} = \mt{b} \}$. 
Such a tight hyper-box can be founded by solving $2N$ linear programs, where $N$ is the width of matrix $\mt{G}$, or can be outer-approximated more efficiently by an iterative method proposed in \cite{scott2016constrained}. 
However, for $\langle \mt{G}\,\text{diag}(\overline{\mtg{\sigma}}), \mt{c}', \mt{A}\,\text{diag}(\overline{\mtg{\sigma}}), \bfz\rangle$ obtained by solving \eqref{eq:minouts}, 
$\li - \bfo, \bfo\ri$ is tight. 
This is because, by the above interpretation, 
  $\li-\overline{\mtg{\sigma}}, \overline{\mtg{\sigma}}\ri$ is the smallest hyper-box that contains $\langle \mtg{\Gamma}, \bfz\rangle \subseteq \{\mtg{\sigma} \in \li -\overline{\mtg{\sigma}}, \overline{\mtg{\sigma}} \ri \mid \mt{A}\mtg{\sigma} = \bfz\} = \mm{S}$.  
Hence  $\li-\overline{\mtg{\sigma}}, \overline{\mtg{\sigma}}\ri$ is also the smallest hyper-box containing $ \mm{S}$.  
As we will see, the tightness of  $\li-\overline{\mtg{\sigma}}, \overline{\mtg{\sigma}}\ri$ plays an important role in the conservatism analysis of Step II. 
\end{rmk}

\begin{rmk}
The cost function $\Vert \vert \mtg{\Gamma}\vert \bfo \Vert_1$ of \eqref{eq:minouts} is the absolute element sum of the matrix $\mtg{\Gamma}$.
One may ask whether minimizing this cost would achieve $\vert \mtg{\Gamma}\vert \bfo \leq \bfo$ whenever possible, even after removing the constraint $\vert \mtg{\Gamma}\vert \bfo \leq \bfo$. 
This is not the case in general though it happens often times. 
If we ignore $\vert \mtg{\Gamma}\vert \bfo \leq \bfo$ in \eqref{eq:minouts} and minimize the Frobenius norm of $\mtg{\Gamma}$ instead, it is equivalent to finding the minimum norm solution of the least square problem defined by $[\mt{G};  \mt{A}]\mtg{\Gamma} = [\mt{G}'; \bfz]$.
If this minimum norm solution also satisfies $\vert\mtg{\Gamma}\vert \bfo \leq \bfo$, then it is a good estimate of the minimizer of \eqref{eq:minout} and can be found more efficiently. 
\end{rmk}

\subsection{Step II: $\mm{CZ} \ominus \mm{CZ}_{\rm s}$ by CG-Rep Manipulation}

We further under-approximate $\mm{CZ} \ominus \mm{CZ}_{\rm s}$ by $\mm{CZ}_{\rm d} \cgeq \langle \mt{G}\,\text{diag}(\bfo - \overline{\mtg{\sigma}}), \mt{c}-\mt{c}', \mt{A}\,\text{diag}(\bfo - \overline{\mtg{\sigma}}), \mt{b}\rangle$. It is tempting to conclude that $\mm{CZ}_{\rm d} = \mm{CZ} \ominus \mm{CZ}_{\rm s}$, but this is not true in general. 
In what follows, we show $\mm{CZ}_{\rm d} \subseteq \mm{CZ} \ominus \mm{CZ}_{\rm s}$ and give a sufficient condition for this under-approximation to be exact. 

The following two propositions will be useful later. 
\begin{prop}\label{prop:boxcase}
Define 
\begin{align}
{\mm{M}} & = \{\mtg{\mu} \in \li-\overline{\mtg{\delta}}, \overline{\mtg{\delta}}\ri \oplus  \li-\overline{\mtg{\sigma}}, \overline{\mtg{\sigma}}\ri \mid  \mt{A}\mtg{\mu} = \mt{b}\},  
\label{eq:mM}\\
{\mm{S}} & = \{\sigma \in  \li-\overline{\mtg{\sigma}}, \overline{\mtg{\sigma}}\ri \mid  \mt{A}\mtg{\sigma} = \bfz\}, 
\label{eq:mS}\\
{\mm{D}} & = \{\mtg{\delta} \in \li-\overline{\mtg{\delta}}, \overline{\mtg{\delta}}\ri \mid \mt{A}\mtg{\delta} = \mt{b}\}. 
\end{align}
Assume that $\li-\overline{\mtg{\sigma}}, \overline{\mtg{\sigma}}\ri$ is the smallest hyper-box that contains $\mm{S}$. 
Then ${\mm{M}} \ominus {\mm{S}} = {\mm{D}}$. 
\end{prop}

\begin{proof}
This is a direct result of Lemma \ref{lem:HmB}. 
Particularly, since $\li-\overline{\mtg{\sigma}}, \overline{\mtg{\sigma}}\ri$ is the smallest hyper-box that contains $\mm{S}$, $\max \pm \mt{e}_i^\top \mm{S} = \pm \overline{\sigma}_i$, where
 $\mt{e}_i$ is the $i^{\rm th}$ natural basis vector and $\overline{\sigma}_i$ is the $i^{\rm th}$  element of $\overline{\mtg{\sigma}}$. 
\end{proof}


\begin{prop}\label{prop:GMNS}
Define $\mm{M}$, $\mm{S}  \subseteq \mathbb{R}^N$ the same as in Proposition \ref{prop:boxcase}, let $\mt{G} \in \mathbb{R}^{n \times N}$ and $\mm{N} : = \mathcal{N}(\mt{A})\cap\mathcal{N}(\mt{G})$. 
Then $\mt{G}\,\mm{M} \ominus \mt{G}\,\mm{S} = \mt{G}\,(\mm{M} \oplus \mm{N} \ominus \mm{S})$. 
\end{prop}

\begin{proof}
By Lemma \ref{lem:Minkowski}, bullets i) and ii) 
\begin{align}
\mt{G}\,(\mm{M} \oplus \mm{N} \ominus \mm{S}) & \subseteq  \mt{G}\,(\mm{M} \oplus \mm{N}) \ominus \mt{G}\,\mm{S} \nonumber \\
& = \mt{G}\, \mm{M} \ominus \mt{G}\, \mm{S}. 
\end{align} 
 It remains to prove that $\mt{G}\,\mm{M} \ominus \mt{G}\,\mm{S} \subseteq \mt{G}\,(\mm{M} \oplus \mm{N} \ominus \mm{S}) $. To this end, let $\mt{x} \in \mt{G}\,\mm{M} \ominus \mt{G}\,\mm{S}$ be arbitrary. Since $\mt{x} \oplus \mt{G}\,\mm{S} \subseteq \mt{G}\,\mm{M}$, we have 
\begin{align}
\forall  \mtg{\sigma} \in \mm{S}: \exists \mtg{\mu}_{\mtg{\sigma}} \in \mm{M}:  \mt{x} + \mt{G}\mtg{\sigma} = \mt{G}\mtg{\mu}_{\mtg{\sigma}}. 
\label{eq:rho}
\end{align}
Now let $\mtg{\sigma}$, $\mtg{\sigma}' \in \mm{S}$ be arbitrary, Eq. \eqref{eq:rho} tells us  
\begin{align}
\mt{A}\mtg{\mu}_{\mtg{\sigma}} & = \mt{A}\mtg{\mu}_{\mtg{\sigma}'} = \mt{b}, \label{eq:Ab}\\
\mt{G}(\underbrace{\mtg{\mu}_{\mtg{\sigma}} - \mtg{\sigma}}_{=:\mtg{\delta}_{\mtg{\sigma}}}) & = \mt{G}(\underbrace{\mtg{\mu}_{\mtg{\sigma}'} - \mtg{\sigma}'}_{=:\mtg{\delta}_{\mtg{\sigma}'}}) = \mt{x}. \label{eq:Gx}
\end{align}
Clearly, $\mtg{\delta}_{\mtg{\sigma}} - \mtg{\delta}_{\mtg{\sigma}'} \in \mm{N}$ by Eqs. \eqref{eq:Ab}, \eqref{eq:Gx} and the fact that $\mtg{\sigma}$, $\mtg{\sigma}' \in \mathcal{N}(\mt{A})$. 
This further implies that  
\begin{align}
\mtg{\delta}_{\mtg{\sigma}} + \mtg{\sigma}' & =  \mtg{\delta}_{\mtg{\sigma}'} - \mtg{\delta}_{\mtg{\sigma}'} + \mtg{\delta}_{\mtg{\sigma}} + \mtg{\sigma}' \nonumber \\
& = \mtg{\mu}_{\mtg{\sigma}'} + \underbrace{(\mtg{\delta}_{\mtg{\sigma}} - \mtg{\delta}_{\mtg{\sigma}'})}_{\in \mm{N}} \in \mm{M} \oplus \mm{N}.  \label{eq:mns}
\end{align}
Since $\mtg{\sigma}' \in \mm{S}$ is arbitrary, Eq. \eqref{eq:mns} implies that 
\begin{align}
& \mtg{\delta}_{\mtg{\sigma}} \oplus \mm{S} \subseteq \mm{M} \oplus \mm{N} \ \ 
\iff \ \  \mtg{\delta}_{\mtg{\sigma}} \in \mm{M} \oplus \mm{N} \ominus \mm{S}. 
\label{eq:final1}
\end{align}
Note that, by Eq. \eqref{eq:rho}, $\mt{x} = \mt{G}(\mtg{\mu}_{\mtg{\sigma}} - \mtg{\sigma}) = \mt{G}\mtg{\delta}_{\mtg{\sigma}}$. Combining this with Eq. \eqref{eq:final1} yields $\mt{x} \in \mt{G}\, (\mm{M} \oplus \mm{N} \ominus \mm{S})$. 
\end{proof}

Now we state the main result of this part. 
\begin{thm}\label{thm:step2}
Let $\mm{CZ} \cgeq \langle \mt{G}, \mt{c}, \mt{A},\mt{b}\rangle$ and $\mm{CZ}_{\rm s} \cgeq \langle  \mt{G}\text{diag}(\overline{\mtg{\sigma}}), \mt{c}', \mt{A}\text{diag}(\overline{\mtg{\sigma}}),\bfz \rangle $, then $\mm{CZ}_{\rm d} \cgeq \langle \mt{G}\,\text{diag}(\bfo - \overline{\mtg{\sigma}}), \mt{c}-\mt{c}', \mt{A}\,\text{diag}(\bfo - \overline{\mtg{\sigma}}), \mt{b}\rangle \subseteq \mm{CZ} \ominus \mm{CZ}_{\rm s}$. 
Further, if $\mm{N}: = \mm{N}(\mt{G}) \cap \mm{N}(\mt{A}) = \{\bfz\}$, we have $\mm{CZ}_{\rm d} = \mm{CZ} \ominus \mm{CZ}_{\rm s}$. 
\end{thm}

\begin{proof}
Note that $\mm{CZ}_{\rm s} = \mt{c}' + \mt{G}\, \mm{S}$ where $\mm{S} = \{\mtg{\sigma} \in \li -\overline{\mtg{\sigma}}, \overline{\mtg{\sigma}} \ri \mid \mt{A}\mtg{\sigma} = \bfz\}$ (see Remark \ref{rmk:wlog}). Also note that $\mm{CZ} = \mt{c} + \mt{G}\, \mm{M}$ where $\mm{M}  = \{\mtg{\mu} \in \li -\bfo, \bfo\ri \mid \mt{A}\mtg{\mu} = \mt{b}\} =  \{\mtg{\mu} \in \li -\bfo+\overline{\mtg{\sigma}}, \bfo - \overline{\mtg{\sigma}}\ri \oplus  \li -\overline{\mtg{\sigma}}, \overline{\mtg{\sigma}} \ri  \mid \mt{A}\mtg{\mu} = \mt{b}\}$. Define $\mm{D} = \{\mtg{\delta} \in \li -\bfo+\overline{\mtg{\sigma}}, \bfo - \overline{\mtg{\sigma}}\ri    \mid \mt{A}\mtg{\delta}= \mt{b}\}$.  We have $\mm{D} = \mm{M}\ominus \mm{S}$ by Remark \ref{rmk:wlog} and Proposition \ref{prop:boxcase}. Also note that 
\begin{align}
& \mm{CZ} \ominus \mm{CZ}_{\rm s} \nonumber \\
= & \ (\mt{c}-\mt{c}') + \mt{G}\,\mm{M} \ominus \mt{G}\,\mm{S} \nonumber \\
= & \ (\mt{c}-\mt{c}') + \mt{G}(\mm{M} \oplus \mm{N} \ominus \mm{S}) & (\text{Proposition \ref{prop:GMNS}})\nonumber \\
\supseteq & \  (\mt{c}-\mt{c}') + \mt{G}(\mm{M} \ominus \mm{S} \oplus \mm{N}) & (\text{Lemma \ref{lem:Minkowski}})\nonumber \\
= & \  (\mt{c}-\mt{c}') +  \mt{G}(\mm{M} \ominus \mm{S}) \oplus \mt{G}\,\mm{N} &  (\text{Lemma \ref{lem:Minkowski}}) \nonumber \\
= & \  (\mt{c}-\mt{c}') +  \mt{G}(\mm{M} \ominus \mm{S})  & (\mm{N}\subseteq \mm{N}(G))\nonumber \\
= & \  (\mt{c}-\mt{c}') + \mt{G}\,\mm{D} =   \mm{CZ}_{\rm d}.  & (\text{Proposition \ref{prop:boxcase}}) 
\label{eq:keyeq}
\end{align}
Note that 
``$\supseteq$'' in Eq. \eqref{eq:keyeq} holds as ``$=$'' if $\mm{N} = \{\bfz\}$. 
\end{proof}
\begin{figure}[h]
  \centering
  \includegraphics[width=2.3in]{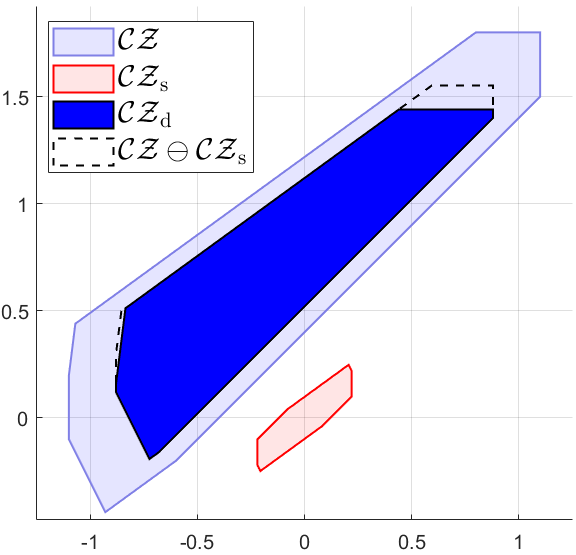}\\
  \caption{Example \ref{eg:diff}. }
\label{fig:eg_diff}
\vspace{-7mm}
\end{figure}

\begin{eg}\label{eg:diff}
With an example, we illustrate that $\mm{CZ} \ominus \mm{CZ}_{\rm s} \neq \mm{CZ}_{\rm d}$ 
in general. 
Define 
\begin{align}
& \mt{G} = 
\left[
\begin{array}{cccc}
1 & 0 & 0 & 0.1 \\
0 & 1 & 0 & 0.8 
\end{array}
\right],   & & \mt{c} = \left[
\begin{array}{cccc}
0  \\
0  
\end{array}
\right], \\
& \mt{A} =  
\left[
\begin{array}{cccc}
-1 & 1 & 0.3 & 1 
\end{array}
\right],  & & \mt{b} = 1, 
\end{align}
and $\overline{\mtg{\sigma}} = [0.2, 0.2, 0.2, 0.2]^\top$. Let $\mm{CZ} \cgeq \langle \mt{G},\mt{c}, \mt{A}, \mt{b}\rangle$ and  $\mm{CZ}_{\rm s}$ be defined as in Proposition \ref{prop:minout}. Fig. \ref{fig:eg_diff} shows that there is a gap between $\mm{CZ} \ominus \mm{CZ}_{\rm s}$ and its under-approximation $\mm{CZ}_{\rm d}$. 
\end{eg}

\begin{rmk}
Since \revise{$\mm{CZ} \ominus \mm{CZ}_{\rm s} \supsetneq \mm{CZ}_{\rm d}$} in general, it is possible that $\mm{CZ}\ominus \mm{CZ}_{\rm s} \neq \emptyset$ but $\mm{CZ}_{\rm d} = \emptyset$. This issue can be mitigated by enforcing the following constraint in \eqref{eq:minout}:  
$ \mtg{\theta} \in \li -\bfo+ \overline{\mtg{\sigma}}, \bfo-\overline{\mtg{\sigma}}\ri$,  $\mt{A}\mtg{\theta} = \mt{b}$, where $\mtg{\theta}$ is a decision variable. 
This extra constraint will ensure that $\mm{CZ}_{\rm d} \neq \emptyset$ whenever possible. 
\end{rmk}

\section{Exactness under A Rich CG-Rep of $\mm{CZ}$}\label{sec:rep}
The CG-Rep of a constrained zonotope is not unique. 
A notable feature of our two-step approach is: the obtained under-approximated difference $\mm{CZ}_{\rm d}$ varies with the CG-Rep of the minuend $\mm{CZ}$. In fact, if the CG-Rep of $\mm{CZ}$ 
is ``rich'' enough, the two-step approach is exact  when $\mm{CZ}\ominus \mm{Z} \neq \emptyset$. 
Such a rich CG-Rep of $\mm{CZ}$ can be constructed as follows. 
Let 
\begin{itemize}
\item[1)] $\mt{H} \in \mathbb{R}^{\ell \times n}$, $\mt{a} \in \mathbb{R}^\ell$ be s.t. $\{\mt{x} \in \mathbb{R}^n \mid \mt{H}\mt{x} \leq \mt{a}\} = \mm{CZ}$;   
\item[2)] $\mt{h}_i^\top \hspace{-1.5mm}\neq \hspace{-0.5mm}\bfz$ be the $i^{\rm th}$ row of $\mt{H}$ and $a_i$ the $i^{\rm th}$ element of $\mt{a}$ (if $\mt{h}_i^\top = \bfz$, then $a_i = 0$ and this row can be removed); 
\item[3)] $r > 0$ be a sufficiently large real number s.t.  $\mm{CZ}_0 := \langle r\,\mt{I}_n, \mt{c} \rangle$ encloses $\mm{CZ}$ for some $\mt{c} \in \mt{c}' + (\mm{CZ}\ominus \mm{Z}) \neq \emptyset$, where $\mt{c}'$ is the center of $\mm{Z}$; 
\item[4)] $\langle \mt{G}, \mt{c}, \mt{A}, \mt{b} \rangle$ be the CG-Rep of $\mm{CZ}_\ell$ obtained via iteratively applying Lemma \ref{lem:cz} iv) to the following intersection operation, which leaves the center $c$ unchanged: 
\begin{align}
\hspace{-6mm} \mm{CZ}_i = \mm{CZ}_{i-1} \cap \{\mt{x} \mid \mt{h}_i^\top \mt{x} \leq a_i\}, \ \ i = 1,2,\dots, \ell. \label{eq:CZi}
\end{align} 
\end{itemize}
Clearly, $\langle \mt{G}, \mt{c}, \mt{A}, \mt{b} \rangle$ is a CG-Rep of set $\mm{CZ}$ because $\langle \mt{G}, \mt{c}, \mt{A}, \mt{b} \rangle \cgeq \mm{CZ}_\ell = \{\mt{x}\mid \mt{H}\mt{x} \leq \mt{a}\} = \mm{CZ}$. 
To be precise, 
\begin{align}
\mt{G} & = [r\, \mt{I}_n, \bfz_{n\times \ell}], \label{eq:G}\\
\mt{A} & = [r\, \mt{H}, \text{diag}(\tfrac{1}{2}\mt{d})], \label{eq:A}\\
d_i & =  a_i - \mt{h}_i^\top \mt{c} + r \Vert \mt{h}_i^\top \Vert_1, \ i= 1,2,\dots, \ell, \label{eq:di}\\
b_i & = \tfrac{a_i - \mt{h}_i^\top \mt{c} - r \Vert \mt{h}_i^\top \Vert_1}{2}, \ i= 1,2,\dots, \ell, \label{eq:bi}
\end{align}
where $d_i$ ($b_i$, resp.) is the $i^{\rm th}$ element of vector $\mt{d}$ ($\mt{b}$, resp.).

In the rest of this section, we use $\langle \mt{G}, \mt{c}, \mt{A}, \mt{b} \rangle$ as the CG-Rep of the minuend $\mm{CZ}$ and show that Step I and Step II are exact when $\mm{CZ}\ominus \mm{Z} \neq \emptyset$. 

The following lemma will be useful. 

\begin{lem}\label{lem:dig}
Assume that $\mm{CZ} \ominus \mm{Z} \neq \emptyset$, then $d_i > 0$ and $[\mt{G};\mt{A}]$ is invertible. 
\end{lem}
\begin{proof}
Recall that $\mm{CZ} = \{\mt{x} \mid \mt{H}\mt{x} \leq \mt{a}\}$, $\mm{Z} = \langle \mt{c}', \mt{G}'\rangle$. Since $\mt{c} \in \mt{c}' + (\mm{CZ}\ominus \mm{Z})$, i.e., $\mt{c}-\mt{c}' + \mm{Z} \subseteq \mm{CZ}$, we have
\begin{align}
a_i \geq \max \mt{h}_i^\top (\mt{c}-\mt{c}' + \mm{Z}) = \mt{h}_i^\top \mt{c} + \Vert \mt{h}_i^\top \mt{G}'\Vert_1. 
\label{eq:aig}
\end{align}
By Eqs. \eqref{eq:di}, \eqref{eq:aig}, $d_i \geq \Vert \mt{h}_i^\top \mt{G}'\Vert_1 +  r \Vert \mt{h}_i^\top \Vert_1$. 
Since $\mt{h}_i \neq \bfz$ and $r >0$ (see bullets 2), 3)), $d_i>0$.  

By Eqs. \eqref{eq:G}, \eqref{eq:A}, $[\mt{G};\mt{A}] = [r\,\mt{I}_n, \bfz; r\,\mt{H}, \text{diag}(\tfrac{\mt{d}}{2})] \in \mathbb{R}^{(n+\ell)\times(n+\ell)}$.
Since $r > 0$ and $d_i > 0$ for $i = 1,2,\dots, \ell$, $[\mt{G};\mt{A}]$ is a triangular matrix with non-zero diagonal entries. Therefore $[\mt{G};\mt{A}]$ is invertible. 
\end{proof}

The following result says that Step I is exact. 
\begin{prop}\label{prop:exact1}
Suppose that  $\mm{CZ} \ominus \mm{Z} \neq \emptyset$, then  \eqref{eq:minouts} is feasible. 
Moreover, let $\mtg{\Gamma}$ be the minimizer of \eqref{eq:minouts}, define $\overline{\mtg{\sigma}} = \vert\mtg{\Gamma}\vert \bfo$ and $\mm{CZ}_{\rm s} \cgeq \langle \mt{G}\,\text{diag}(\overline{\mtg{\sigma}}), \mt{c}', \mt{A}\,\text{diag}(\overline{\mtg{\sigma}}),\bfz \rangle$, then $\mm{CZ} \ominus \mm{CZ}_{\rm s} = \mm{CZ} \ominus \mm{Z}$. 
\end{prop}

\begin{proof}
By  Lemma \ref{lem:dig}, there is a unique $\mtg{\Gamma}$ satisfying the equality constraint $[\mt{G}; \mt{A}]\mtg{\Gamma}= [\mt{G}'; \bfz]$, i.e., 
\begin{align}
\mtg{\Gamma} & = \big[\tfrac{1}{r} \mt{G}';  -\text{diag}(\tfrac{\mt{d}}{2})^{-1} \mt{H} \mt{G}' \big], 
\end{align} 
In what follows, we show that, assuming that  $\langle r\,\mt{I}_n, \mt{c}\rangle \supseteq \mm{CZ}$, the constraint  
$\vert\mtg{\Gamma}\vert \bfo \leq \bfo$ holds automatically. 
Note that 
\begin{align}
\overline{\mtg{\sigma}} = \vert\mtg{\Gamma}\vert \bfo = \Big[\tfrac{1}{r}\vert \mt{G}'\vert \bfo ; \big\vert\text{diag}(\tfrac{\mt{d}}{2})^{-1}\mt{H}\mt{G}'\big\vert \bfo \Big] . 
\end{align}
\begin{itemize}
\item[i)]For $j = 1,2\dots, n$, $\overline{\sigma}_j = \tfrac{1}{r}\Vert \mt{e}_j^\top \mt{G}' \Vert_1$ where $\mt{e}_j$ is the $j^{\rm th}$ natural basis. 
Since $\langle r\,\mt{I}_n, \mt{c}\rangle \supseteq \mm{CZ} \supseteq \mt{c}-\mt{c}' + \mm{Z} = \langle \mt{G}', \mt{c}'\rangle$,  
we have
\begin{align}
& \, \mt{e}_j^\top \mt{c}' + \Vert \mt{e}_j^\top \mt{G}' \Vert_1 = \, \max \mt{e}_j^\top (\mt{c}-\mt{c}'+\mm{Z}) \nonumber \\
 \leq &   \, \max \mt{e}_j^\top \langle r\, \mt{I}_n, \mt{c}\rangle
=  \, \mt{e}_j^\top \mt{c}' + r. 
\end{align}
Therefore  $\Vert \mt{e}_j^\top \mt{G}' \Vert_1 \leq r$ and $\overline{\sigma}_j \leq 1$. 
\item[ii)]For $i = 1,2,\dots, \ell$, $\overline{\sigma}_{n+i} = \tfrac{2}{d_i} \Vert \mt{h}_i^\top \mt{G}' \Vert_1 = \tfrac{2 \Vert \mt{h}_i^\top \mt{G}' \Vert_1}{\Vert \mt{h}_i^\top \mt{G}' \Vert_1 + r\Vert \mt{h}_i^\top \Vert_1}$. Again, since 
$\langle r\,\mt{I}_n, \mt{c}\rangle  \supseteq \mt{c}-\mt{c}' + \mm{Z} = \langle \mt{G}', \mt{c}'\rangle$,  we have 
\begin{align}
&  \, \mt{h}_i^\top \mt{c}  +  \Vert \mt{h}_i^\top \mt{G}' \Vert_1 = \, \max \mt{h}_i^\top (\mt{c}-\mt{c}'+\mm{Z}) \nonumber \\
 \leq &  \, \max \mt{h}_i^\top \langle r\, \mt{I}_n, \mt{c}\rangle = \, \mt{h}_i^\top \mt{c}  + r\Vert \mt{h}_i^\top \Vert_1. 
\end{align}
Therefore  $\Vert \mt{h}_i^\top \mt{G}' \Vert_1 \leq r\Vert \mt{h}_i^\top \Vert_1$ and $\overline{\sigma}_{n+i} \leq 1$. 
\end{itemize}
So far we have proved that $\mtg{\Gamma}$ is the unique feasible solution (hence the minimizer) of \eqref{eq:minouts}. 

It is known from Lemma \ref{lem:HmB} that $\mm{CZ} \ominus \mm{Z} = \mm{CZ} \ominus \{\mt{x} \mid \mt{H}\mt{x} \leq \mt{a}_{\rm s}\}$, where $\mt{a}_{\rm s} \in \mathbb{R}^\ell$ and its $i^{\rm th}$ element $a_{{\rm s}, i} = \mt{h}_i^\top \mt{c}' + \Vert \mt{h}_i^\top \mt{G}'\Vert_1$. 
Note that, for $i = 1,2,\dots, \ell$,  
\begin{align}
&  \, \max \mt{h}_i^\top \mm{CZ}_{\rm s} \nonumber \\
= &  \, \max \mt{h}_i^\top \{\mt{G}\mtg{\sigma} + \mt{c}' \mid \mtg{\sigma} \in \li -\overline{\mtg{\sigma}}, \overline{\mtg{\sigma}}\ri, \mt{A}\mtg{\sigma} = \bfz \} \nonumber \\
= &   \, \mt{h}_i^\top \mt{c}' + \max \left\{ r\,\mt{h}_i^\top \mtg{\sigma}_{1:n} \bigg\vert 
\begin{array}{l}
\mtg{\sigma} \in \li -\overline{\mtg{\sigma}}, \overline{\mtg{\sigma}}\ri, \ \  r\mt{H}\mtg{\sigma}_{1:n} =\\
 -\text{diag}(\tfrac{\mt{d}}{2}) \mtg{\sigma}_{n+1:n+\ell} 
\end{array}
 \right \}\nonumber \\
\leq  &  \,   \mt{h}_i^\top \mt{c}' + \max \{ - \tfrac{d_i}{2} \sigma_{n+i} \mid \mtg{\sigma} \in \li -\overline{\mtg{\sigma}}, \overline{\mtg{\sigma}}\ri \}\nonumber \\
= &  \, \mt{h}_i^\top \mt{c}' + \tfrac{d_i}{2} \overline{\sigma}_{n+i} \nonumber \\
= & \,   \mt{h}_i^\top \mt{c}'  + \Vert \mt{h}_i^\top \mt{G}' \Vert_1 = a_{{\rm s}, i},
\label{eq:incl}
\end{align}
where $\mtg{\sigma}_{1:n}$ (and $\mtg{\sigma}_{n+1:n+\ell}$, resp.) is a vector that consists of the first $n$ elements (and the last $\ell$ elements, resp.) of $\mtg{\sigma}$. 
By Eq. \eqref{eq:incl}, $\mm{CZ}_{\rm s} \subseteq \{\mt{x} \mid \mt{H}\mt{x} \leq \mt{a}_{\rm s}\}$. 
Together with the fact that $\mm{Z} \subseteq \mm{CZ}_{\rm s}$, we have $\mm{CZ} \ominus \mm{CZ}_{\rm s} = \mm{CZ} \ominus \mm{Z}$. 
\end{proof}

The following result says that Step II is exact. 
\begin{prop}\label{prop:exact2}
Let $\mm{CZ}_{\rm s} \cgeq \langle \mt{G}\,\text{diag}(\overline{\mtg{\sigma}}), \mt{c}', \mt{A}\,\text{diag}(\overline{\mtg{\sigma}}),\bfz \rangle$ 
and 
$\mm{CZ}_{\rm d} \cgeq \langle \mt{G}\,\text{diag}(\bfo - \overline{\mtg{\sigma}}), \mt{c}-\mt{c}', \mt{A}\,\text{diag}(\bfo - \overline{\mtg{\sigma}}),\mt{b} \rangle$, 
where $\overline{\mtg{\sigma}}$ is defined the same as in Proposition \ref{prop:exact1}, then $\mm{CZ} \ominus \mm{CZ}_{\rm s} = \mm{CZ}_{\rm d}$. 
\end{prop}

\begin{proof}
By  Lemma \ref{lem:dig}, $\mm{N}(\mt{G}) \cap \mm{N}(\mt{A}) = \mm{N}([\mt{G};\mt{A}]) = \{\bfz\}$. 
By Theorem \ref{thm:step2}, $\mm{CZ} \ominus \mm{CZ}_{\rm s} = \mm{CZ}_{\rm d}$. 
 \end{proof}
 
Although $\mm{CZ}\ominus \mm{Z}  \neq \emptyset$ is assumed in Propositions \ref{prop:exact1}, \ref{prop:exact2}, 
the result $\mm{CZ}_{\rm d}$ returned by the two-step approach is exact in the following sense, regardless of this assumption. 

\begin{thm}\label{thm:exact}
For any constrained zonotope $\mm{CZ}$, if we construct its CG-Rep following steps 1)-4), then the followings hold: 
\begin{itemize}
\item[i)] if $\mm{CZ}\ominus \mm{Z} = \emptyset$, either $\mm{CZ}_{\rm d} = \emptyset$ or \eqref{eq:minouts} is infeasible; 
\item[ii)] if $\mm{CZ}\ominus \mm{Z} \neq \emptyset$, $\mm{CZ}_{\rm d} = \mm{CZ}\ominus \mm{Z}$. 
\end{itemize}
\end{thm}

\begin{proof}
Bullet ii) follows from Proposition  \ref{prop:exact1} and Proposition \ref{prop:exact2}. 
For bullet i), if \eqref{eq:minouts} is feasible, then $\mm{CZ}_{\rm d} = \emptyset$ because $\mm{CZ}_{\rm d}\subseteq \mm{CZ}\ominus \mm{Z} = \emptyset$ by Theorem \ref{thm:step2}. 
\end{proof}
 


\begin{rmk}\label{rmk:hard}
\revise{The run time of our two-step approach is polynomial in the input size. Thus, by Proposition \ref{prop:hard}, this approach cannot be exact \emph{in general} unless P $=$ NP. }
However, as stated by Theorem \ref{thm:exact}, the two-step approach still achieves exactness \emph{in special cases} where $\mm{CZ}$'s CG-Rep is not the most ``compact'' one. 
Note that,  such non-compact CG-Rep is constructed from the H-Rep of $\mm{CZ}$, whose complexity may, in the worst case, be exponential in that of $\mm{CZ}$'s most compact CG-Rep. 
Therefore, the exactness results are not surprising because in this case, the two-step approach bypasses the high-complexity step of computing $\mm{CZ}$'s H-Rep. 
This is also consistent with the fact that 
Minkowski-subtracting a zonotope from a polytope in its H-Rep is easy \cite{kolmanovsky1998theory}. 
However, our results in Sec. \ref{sec:rep} may open the direction of incrementally enriching the minuend's CG-Rep to  achieve exact results quickly whenever possible. 
\end{rmk}



\section{Backward Reachable Set Computation for Nonlinear Systems} \label{sec:brs_nl}
In this section, we use our two-step approach from Sec. \ref{sec:Mdiff} to develop BRS under-approximation algorithms  for system \eqref{eq:sysnl}, as promised in Sec. \ref{sec:prob}. In particular, to incorporate the error introduced by sequential linearization, we require $\underline{\mm{X}}_{k-1} \ominus \mm{L}$ to be reachable from $\underline{\mm{X}}_k$ under the linearized dynamics, where $\underline{\mm{X}}_{k-1}$ ( $\underline{\mm{X}}_k$, resp.) are the $k-1^{\rm st}$ ($k^{\rm th}$, resp.) under-approximated BRS and $\mm{L}$ contains all values of the linearization error over $\underline{\mm{X}}_k$. 
As mentioned at the end of Sec. \ref{sec:prob}, the challenge is to approximate $\mm{L}$ without knowing $\underline{\mm{X}}_k$ a priori. 
To resolve this issue, we explore the following two strategies. 
\begin{itemize}
    \item[A)] Scaling method
    : we incrementally enlarge $\mm{L}$ by a scaling factor until i) $\underline{\mm{X}}_{k-1} \ominus \mm{L}$ is reachable from $\underline{\mm{X}}_k$ under the linear dynamics, and ii) $\mm{L}$ encompasses all the values of the linearization error in $\underline{\mm{X}}_k$. Here, each $\underline{\mm{X}}_k$ is a constrained zonotope. 
    \item[B)] Splitting method: we fix $\mm{L}$ and split $\underline{\mm{X}}_{k-1}$ into finitely many smaller sets, i.e.,  $\underline{\mm{X}}_{k-1} = \bigcup_i  \underline{\mm{X}}_{k-1}^i$. The nonlinear system is linearized for each $\underline{\mm{X}}_{k-1}^i$, and the splitting procedure terminates when the linearization error in each $\underline{\mm{X}}_{k}^i$, from where  $\underline{\mm{X}}_{k-1}^i \ominus \mm{L}$ is reachable under the $i^{\rm th}$ linear dynamics, are contained by $\mm{L}$. In this case, $\underline{\mm{X}}_k  = \bigcup_i  \underline{\mm{X}}_{k}^i$ and is represented by the collection of the CG-Reps of constrained zonotopic sets $\underline{\mm{X}}_{k}^i$.  
\end{itemize}
The scaling method is in principle similar to the approach  proposed in \revise{ \cite{althoff2013reachability}}, and the splitting method borrows the idea from \cite{althoff2008reachability}, where a similar splitting procedure is developed for zonotopes to control the linearization error in forward reachability analysis. Unique to our implementation is the use of our efficient Minkowski-difference computation techniques tailored to constrained zonotopes. 
While the scaling method is more suitable for computation with more steps (i.e., larger $k$) in a convex safe set $\mm{X}_{\rm safe}$, the splitting method can better capture the shape of a  nonconvex set $\mm{X}_k$ and is more suitable when the safe set $\mm{X}_{\rm safe}$ is also nonconvex, because $\underline{\mm{X}}_k$ is represented as a collection of constrained zonotopes in the latter method. 
In what follows, we present the detailed  algorithms for the two methods above. 


\subsection{Scaling Method}

\begin{algorithm}[htp]
 \caption{$\underline{\mm{X}}_k = $ ScalingBRS($\underline{\mm{X}}_{k-1}, \mt{f}, \mm{U},\mm{W},  \mm{X}_{\rm safe}$)}
 \begin{algorithmic}[1]
 \REQUIRE Constrained zonotope $\underline{\mm{X}}_{k-1}$; 
 System's vector field $\mt{f}$; Control input set $\mm{U}$; Disturbance set $\mm{W}$; 
 Safe set $\mm{X}_{\rm safe}$ 
 \ENSURE Constrained zonotope $\underline{\mm{X}}_k\subseteq  Pre(\underline{\mm{X}}_{k-1})$
  \STATE $\widetilde{\mt{z}} \leftarrow  center(\underline{\mm{X}}_{k-1} \times \mm{U})$ 
  \STATE $[\widetilde{\mt{A}}, \widetilde{\mt{B}}] \leftarrow  linearize(\widetilde{\mt{z}}, \mt{f})$; \ $\widetilde{\mm{L}} \leftarrow \{\mt{f}(\widetilde{\mt{z}}) - [ \widetilde{\mt{A}},  \widetilde{\mt{B}}] \widetilde{\mt{z}}\}$
  \STATE $\widetilde{\mm{Z}}_k  \leftarrow  Pre_{\mt{x}, \mt{u}}(\underline{\mm{X}}_{k-1}, \widetilde{\mt{A}},  \widetilde{\mt{B}}, \mm{U}, \mm{W}, \widetilde{\mm{L}}, \mm{X}_{\rm safe})$
  \STATE $ \mt{z} ^\ast \leftarrow center(\widetilde{\mm{Z}}_k)$ 
  \STATE $[\mt{A}, \mt{B}] \leftarrow linearize(  \mt{z}^\ast, \mt{f})$; \ $\mm{L} \leftarrow LE( \mt{z}^\ast, \mt{f},  \widetilde{\mm{Z}}_k)$
  \STATE $\mm{Z}_k \leftarrow  Pre_{\mt{x}, \mt{u}}(\underline{\mm{X}}_{k-1},\mt{A}, \mt{B}, \mm{U}, \mm{W},  \mm{L}, \mm{X}_{\rm safe})$
  \WHILE{ $LE( \mt{z}^\ast, \mt{f}, \mm{Z}_k) \not\subseteq  \mm{L}$}
  \STATE Enlarge  $\mm{L}$  by a factor $\alpha$
  \STATE $\mm{Z}_k \leftarrow  Pre_{\mt{x}, \mt{u}}(\underline{\mm{X}}_{k-1}, \mt{A}, \mt{B}, \mm{U}, \mm{W},  \mm{L},  \mm{X}_{\rm safe})$
  \ENDWHILE
 \RETURN $\underline{\mm{X}}_k \leftarrow Proj_{\mt{x}}(\mm{Z}_k)$
 \end{algorithmic}
 \label{alg:scaling}
 \end{algorithm}
Algorithm \ref{alg:scaling} details the scaling method. In this algorithm, 
 we first linearize the system at the geometric center $\widetilde{z}$ of \revise{the interval closure of $\underline{\mm{X}}_{k-1}\times \mm{U}$}. 
The function $center(\mm{Z})$ returns $\tfrac{1}{2}(\underline{\mt{z}}  + \overline{\mt{z}})$, where $\underline{\mt{z}}$ and $\overline{\mt{z}}$ are the lower and upper limits of the smallest hyper-box that contains set $\mm{Z}$. That is, 
\begin{align}
     \underline{z}_i & =  \min \mt{e}_i^\top \mm{Z}, \ \ \ \  \overline{z}_i  = \max \mt{e}_i^\top \mm{Z}.
     \label{eq:zbars}
\end{align}
where $\underline{z}_i$ and $\overline{z}_i$ are the $i^{\rm th}$ elements of $\underline{\mt{z}}$ and $\overline{\mt{z}}$, respectively. 
Suppose that $\mm{Z}$ is a constrained zonotope, executing the function $center$ amounts to solving $2n$ linear programs. 
On line 2, the function $linearize(\mt{z} , \mt{f})$ linearizes the vector field $\mt{f}$ at point $\mt{z}$ and returns the  matrices that define the obtained linear system, i.e., 
\begin{align}
    &  \mt{A}  = \frac{\partial{\mt{f}(\mt{x},\mt{u})}}{\partial{\mt{x}}}\bigg|_{[\mt{x};\mt{u}] = \mt{z}  }, \ \ \  \mt{B}  = \frac{\partial{\mt{f}(\mt{x},\mt{u})}}{\partial{\mt{u}}}\bigg|_{[\mt{x};\mt{u}] = \mt{z}  }. 
    \label{eq:AB_lin}
\end{align}
On line 3, we compute a set $\widetilde{\mm{Z}}_k$ of state-input vectors $[\mt{x};\mt{u}]$ such that $\underline{\mm{X}}_{k-1}$ is reached from state $\mt{x}$ under control $\mt{u}$ and the obtained linear dynamics. To be precise, $Pre_{\mt{x}, \mt{u}}(\mm{X}, \mt{A}, \mt{B}, \mm{U}, \mm{W}, \mm{L}, \mm{X}_{\rm safe})$ is defined to be:
\begin{align}
\{[\mt{x};\mt{u}] \in \mm{X}_{\rm safe}\times \mm{U} \mid& \mt{A}\mt{x} + \mt{B}\mt{u} \in \mm{X}\ominus (\mm{L} \oplus \mm{W})\}. \label{eq:Zk}
\end{align}
Suppose that $\mm{X}_{\rm safe}$ is a polytope, then  $ Pre_{\mt{x},\mt{u}}$ is a constrained zonotope whose CG-Rep can be obtained using Lemma \ref{lem:cz} and our two-step approach for Minkowski-difference computation. Note that, in line 3, $ \widetilde{\mm{Z}}_k$ is computed without considering any linearization error (i.e., $\widetilde{\mm{L}} = \{\mt{f}(\widetilde{\mt{z}} ) - [ \widetilde{\mt{A}},  \widetilde{\mt{B}}] \widetilde{\mt{z}}\}$ is a singleton set). The purpose of this step is to find a better point $\mt{z}^\ast = center(\widetilde{\mm{Z}}_k)$ to linearize the system at (line 4). 
Then on lines 5-6, we linearize the system at $\mt{z}^\ast$ and recompute a set $\mm{Z}_k$ with the latest linear system and a set $ \mm{L}  = LE( \mt{z}^\ast, \mt{f},  \widetilde{\mm{Z}}_k)$ that contains all possible values of the linearization error. Particularly, $LE(\mt{z}^\ast, \mt{f}, \mm{Z})$ is a zonotope that encloses the following set of Lagrange remainders over set $\mm{Z}$: 
\begin{align}
    & \hspace{-1mm}\mt{f}(\mt{z}^\ast) - [\mt{A},  \mt{B}] \mt{z}^\ast + \nonumber \\
    & \hspace{-2mm}\left\{\hspace{-0.5mm}\mt{L} \in \mathbb{R}^n \,\bigg\vert\hspace{-1.5mm}
    \begin{array}{l}
         L_i = \tfrac{1}{2}(\mt{z}- \mt{z}^\ast)^\top \tfrac{\partial^2 f_i}{\partial \mt{z}^2}(\mtg{\xi}_i)(\mt{z}- \mt{z}^\ast), \ \mt{z} \in \mm{Z}  \\
         \mtg{\xi}_i = \lambda_i\mt{z}^\ast + (1-\lambda_i)\mt{z}, \ \lambda_i \in \li 0,1\ri
    \end{array}\hspace{-2.2mm}\right\}
\end{align}
where $\mt{A}$, $\mt{B}$ are given by Eq. \eqref{eq:AB_lin} evaluated at $\mt{z}^\ast$, and $L_i$, $f_i$ are the $i^{\rm th}$ elements of $\mt{L}$ and $\mt{f}$, respectively.  
The set $LE(\mt{z}^*, \mt{f}, \mm{Z})$ can be computed as a hyper-box using interval analysis techniques (e.g., see  \cite{althoff2008reachability}). If $\mm{L}$ encloses all possible values of the linearization error in set $\mm{Z}_k$, then $\mm{X}_{k-1}$ can be reached from   $Proj_{\mt{x}}(\mm{Z}_k) := \{\mt{x} \mid [\mt{x}; \mt{u}] \in \mm{Z}_k \}$ under the nonlinear dynamics\footnote{The projection step amounts to a linear transformation and is easy for CG-Reps by Lemma \ref{lem:cz}, bullet i).}. In that case, we return $\underline{\mm{X}}_k = Proj_{\mt{x}}(\mm{Z}_k)$. 
Otherwise we incrementally enlarge $\mm{L}$ by a factor $\alpha$ and recompute $\mm{Z}_k$ until the linearization error set $LE( \mt{z}^\ast, \mt{f}, \mm{Z}_k)$ is enclosed by $ \mm{L} $.

\subsection{Splitting Method} \label{subsec:splitting}
The real backward reachable set $\mm{X}_k$ may not be convex due to the nonlinearity of the system's vector field $\mt{f}$ or the nonconvexity of the safe set $\mm{X}_{\rm safe}$. 
Therefore, the scaling method can be conservative because we under-approximate a potentially nonconvex set $\mm{X}_k$ with a  constrained zonotope $\underline{\mm{X}}_k$,
which is a convex set. In this part, we present another approach called the splitting method, 
where the linearization error set $\mm{L}$ is defined to have a prescribed (i.e., fixed) size, and $\underline{\mm{X}}_k$ is represented as the \emph{union} of a finite collection  $\{\underline{\mm{X}}_{k}^i\}$ of constrained zonotopes, so that the linearization error in each $\underline{\mm{X}}_{k}^i$ is over-approximated by the prescribed $\mm{L}$. 
Since $\underline{\mm{X}}_k = \bigcup_i \underline{\mm{X}}_k^i$ is not necessarily convex in the splitting method, it serves as a less conservative (i.e., larger) under-approximation of $\mm{X}_k$. 


The following proposition provides a rigorous way to split one constrained zonotope into two, so that the linearization error can be evaluated over the two smaller sets separately. 
\begin{prop}\label{prop:splitting}
Assume that $\mm{CZ} = \langle \mt{G}, \mt{c}, \mt{A}, \mt{b} \rangle$, where 
$ \mt{G} = [\mt{g}_1, \mt{g}_2, \cdots, \mt{g}_N],
    \mt{A}  = [\mt{a}_1, \mt{a}_2, \cdots, \mt{a}_N]$. 
Then set $\mm{CZ}$ can be split into $\mm{CZ}_1^i = \langle \mt{G}_1, \mt{c}_1, \mt{A}_1, \mt{b}_1 \rangle$  and $\mm{CZ}_2^i = \langle \mt{G}_2, \mt{c}_2, \mt{A}_2, \mt{b}_2 \rangle$ along  $\mt{g}_i$, i.e., $\mm{CZ} = \mm{CZ}_1^i \cup \mm{CZ}_2^i$ where 
\begin{align}
    \mt{G}_1 = \mt{G}_2 &= [\mt{g}_1, \mt{g}_2, \cdots,\tfrac{1}{2}\mt{g}_i, \cdots, \mt{g}_N]  \\
    \mt{A}_1 = \mt{A}_2 &= [\mt{a}_1, \mt{a}_2, \cdots, \tfrac{1}{2}\mt{a}_i, \cdots, \mt{a}_N]  \\
    \mt{b}_1 &= \mt{b} - \tfrac{1}{2}\mt{a}_i \quad
    \mt{c}_1 = \mt{c} + \tfrac{1}{2}\mt{g}_i \\
    \mt{b}_2 &= \mt{b} + \tfrac{1}{2}\mt{a}_i \quad
    \mt{c}_2 = \mt{c} - \tfrac{1}{2}\mt{g}_i 
\end{align}
\end{prop}
\begin{proof}
The set $\mm{CZ}$ can be written as
\begin{align}
    & \mm{CZ} = \{ \mt{G\theta+c} \mid |\theta_j| \leq1, j=1,2,\cdots,N, \mt{A\theta}=\mt{b}\}. 
    \end{align}
where $\mt{\theta} = [\theta_1; \theta_2; \cdots;\theta_N] \in \mathbb{R}^{N}$. Define 
    \begin{align}
    \hspace{-3mm}& \hspace{-3mm}\mm{CZ}_1^i \hspace{-0.5mm}=\hspace{-0.5mm} \{ \mt{G\theta+c} \hspace{-0.5mm}\mid \hspace{-0.5mm}\theta_i \in \li 0,1 \ri,|\theta_j| \leq1, j \neq i, \mt{A\theta}=\mt{b}\}, \\
    & \hspace{-3mm} \mm{CZ}_2^i \hspace{-0.5mm}=\hspace{-0.5mm} \{ \mt{G\theta+c} \hspace{-0.5mm}\mid \hspace{-0.5mm}\theta_i \in \li -1,0 \ri,|\theta_j| \leq1, j \neq i, \mt{A\theta}=\mt{b}\} . 
\end{align}
Apparently, $\mm{CZ} = \mm{CZ}_1^i \cup \mm{CZ}_2^i$. 
To show that $\mm{CZ}_1^i $ $ = \langle \mt{G}_1, \mt{c}_1, \mt{A}_1, \mt{b}_1 \rangle$, let  $\mtg{\theta}$ be such that $\theta_i \in \li 0,1 \ri$, $|\theta_j| \leq1$ for all $j \neq i$ and $\mt{A\theta}=\mt{b}$. 
Define $\mu_i = 2\theta_i - 1$ and  $\mt{\hat{\theta}} = [\theta_1; \theta_2; \cdots;\theta_{i-1}; \mu_i; \theta_{i+1}; \cdots; \theta_N]$, then 
\begin{align}
\mt{G\theta+c}  & = \mt{c} + \textstyle\sum_{j=1, j \neq i}^N \mt{g}_j\theta_j + \mt{g}_i(\tfrac{1}{2}\mu_i + \tfrac{1}{2}) \nonumber \\
    &= (\mt{c}+\tfrac{1}{2}\mt{g}_i) + \textstyle\sum_{j=1, j \neq i}^N \mt{g}_j\theta_j + \tfrac{1}{2}\mt{g}_i\mu_i  \nonumber \\
    &= \mt{G}_1 \mt{\hat{\theta}} + \mt{c}_1
\end{align}
In addition,
\begin{align}
    \mt{A\theta} = \mt{b} & \iff \textstyle\sum_{j=1}^N \mt{a}_j \theta_j = \mt{b} \nonumber \\
    & \iff \textstyle\sum_{j=1,j \neq i}^N \mt{a}_j \theta_j + \mt{a}_i(\tfrac{1}{2}\mu_i+\tfrac{1}{2}) = \mt{b} \nonumber \\
    & \iff \mt{A}_1\mt{\hat{\theta}} = \mt{b}_1
\end{align}
Since $|\mu_i| = |2\theta_i - 1| \leq 1$, we have $\hat{\mtg{\theta}} \in \li -\bfo, \bfo\ri$. Therefore, 
\begin{align}
    \mm{CZ}_1^i &= \{ \mt{G_1\hat{\theta}+c_1}\mid \hat{\mtg{\theta}} \in \li -\bfo, \bfo\ri ,\mt{A_1\hat{\theta}}=\mt{b_1}\} \nonumber \\
    &= \langle \mt{G}_1, \mt{c}_1, \mt{A}_1, \mt{b}_1 \rangle
\end{align}
Similarly, to show that $\mm{CZ}_2^i = \langle \mt{G}_2, \mt{c}_2, \mt{A}_2, \mt{b}_2 \rangle$, let $\mtg{\theta}$ be such that $\theta_i \in \li -1,0 \ri$, $|\theta_j| \leq1$ for all $j \neq i$ and $\mt{A\theta}=\mt{b}$. The above argument follows by setting 
$\mu_i = 2\theta_i + 1$. 
\end{proof}
Fig. \ref{fig:prop8} shows the two sets $\mm{CZ}_1$ (green) and  $\mm{CZ}_2$ (red)  obtained by splitting $\mm{CZ}$ (black contour) using Proposition \ref{prop:splitting}. Note that $\mm{CZ}_1$ overlaps with $\mm{CZ}_2$.  As we will see, this overlap helps to reduce the conservatism in the BRS computation. 


\begin{figure}[h]
  \centering
  \includegraphics[width=2.5in]{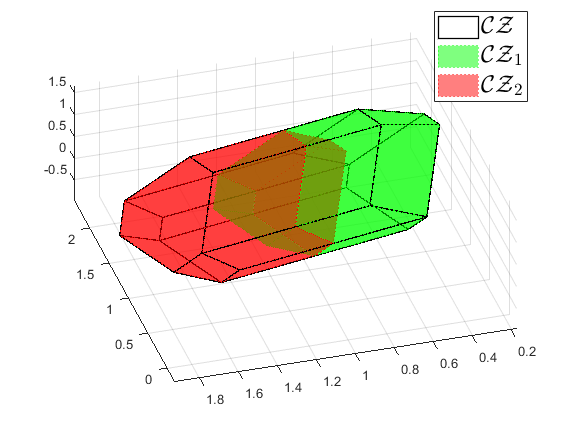}\\
  \caption{Result for Proposition \ref{prop:splitting}: $\mm{CZ} = \mm{CZ}_1 \cup \mm{CZ}_2$}
\label{fig:prop8}
\vspace{-4mm}
\end{figure}

Based on Proposition \ref{prop:splitting}, a detailed algorithm that implements the splitting method is given in Algorithm \ref{alg:split}. 
The output $\{\underline{\mm{X}}_k^\ell\}$ and the input $\{\underline{\mm{X}}_{k-1}^i\}$ are finite collections of constrained zonotopes s.t. $\underline{\mm{X}}_{k-1} = \bigcup_i \underline{\mm{X}}_{k-1}^i$ and $\underline{\mm{X}}_k = \bigcup_\ell \underline{\mm{X}}_k^\ell$ are the under-approximation of the $k-1^{\rm st}$ and the $k^{\rm th}$ BRSs, respectively. 
The sub-procedures $center$, $linearize$ and  $LE$ are the same as in the scaling method. However,  $Pre_{\mt{x},\mt{u}}(\mm{X}, \mt{A}, \mt{B}, \mm{U}, \mm{W}, \mm{L})$, which computes the extended constrained zonotope $\mm{Z}$, has a slightly different definition, i.e.,  
\begin{align}
\mm{Z} = \{[\mt{x};\mt{u}] \mid \mt{A}\mt{x} + \mt{B}\mt{u} \in \mm{X}\ominus (\mm{L} \oplus \mm{W}),
\mt{u} \in \mm{U} \}.  \label{eq:Zk_nosafe}
\end{align}
Note that, in Eq. \eqref{eq:Zk_nosafe},   $\mt{x} \in \mm{X}_{\rm safe}$ is not enforced as in the scaling method. 

 Algorithm \ref{alg:split} is briefly explained below. For each $\underline{\mm{X}}_{k-1}^i$ from the input collection, we construct a collection $\mm{C}_k^i$ of constrained zonotopes, the union of which is contained by $Pre(\underline{\mm{X}}_{k-1}^i)$. 
 To obtain $\mm{C}_k^i$, we compute a set $\mm{Z}_k^i$ of the state-input pairs $[\mt{x}; \mt{u}]$ using linearization (lines 2-5). These steps are the same as those in the scaling method except that the linearization error set $\mm{L}$ is now defined by a prescribed error bound $\bar{\mt{L}}$. 
 If all possible values of the linearization error over $\mm{Z}_k^i$ are contained by $\mm{L}$, it follows that  $Proj_{\mt{x}}(\mm{Z}_{k}^i) \cap \mm{X}_{\rm safe}  \subseteq Pre(\underline{\mm{X}}_{k-1}^i)$ and  $\mm{C}_k^i$ is given as in line 7. Otherwise we split $\underline{\mm{X}}_{k-1}^i$ into two smaller constrained zonotopes $\underline{\mm{X}}_{k-1,1}^{i,j^*}$,   $\underline{\mm{X}}_{k-1,2}^{i,j^*}$ and compute the BRSs for them by calling SplittingBRS recursively. 
 Note that, for the first case (i.e., line 7), each $\mm{C}_k^i$ may still contain multiple sets when $\mm{X}_{\rm safe}$ is nonconvex. 
 For example, if $\mm{X}_{\rm safe} = \bigcup_p \mm{H}_p$ is the union of finitely many polytopes $\mm{H}_p$, the collection $\mm{C}_k^i$ will consist of $Proj_{\mt{x}}(\mm{Z}_{k}^i) \cap \mm{H}_p$ for all $p$. Each $Proj_{\mt{x}}(\mm{Z}_{k}^i) \cap \mm{H}_p$ is a constrained zonotope, whose CG-Rep can be obtained by Lemma \ref{lem:cz}. For details, see \cite{raghuraman2020set}. Finally, if $\mt{f}$ is twice continuously differentiable, $LE({\mt{z}}^*, \mt{f}, \mm{Z})$ will converge to a singleton set as $\mm{Z}$ does. This ensures that the recursion will terminate after sufficiently many splittings.

\begin{algorithm}[H]
 \caption{$\{\underline{\mm{X}}_k^\ell\}$ = SplittingBRS($\{\underline{\mm{X}}_{k-1}^i\}, \mt{f}, \mm{U}, \mm{W}, \mt{\bar{L}}, \mm{X}_{\rm safe}$)}
 \begin{algorithmic}[1]
 \REQUIRE  A collection $\{\underline{\mm{X}}_{k-1}^i\}$ of constrained zonotopes; 
 System's vector field $\mt{f}$; Control input set $\mm{U}$; Disturbance set  $\mm{W}$; Safe set $\mm{X}_{\rm safe}$; Admissible linearization error  $\mt{\bar{L}} \in \mathbb{R}^n$
 \ENSURE  A collection $\{\underline{\mm{X}}_{k}^\ell\}$ of constrained zonotopes s.t. $\bigcup_\ell \underline{\mm{X}}_k^\ell \subseteq Pre(\bigcup_i\underline{\mm{X}}_{k-1}^i)$
 \FOR{each $ \mm{X}_{k-1}^i$}
 \STATE $\mt{z}^* \leftarrow center(\underline{\mm{X}}_{k-1}^i \times \mm{U})$
 \STATE $[\mt{A}, \mt{B}] \leftarrow linearize(\mt{z}^*, \mt{f})$
 \STATE $\mm{L} \leftarrow \langle \text{diag}(\mt{\bar{L}}), \mt{f}(\mt{z}^\ast) -[\mt{A}, \mt{B}]\mt{z}^\ast \rangle$
 \STATE $\mm{Z}_k^i \leftarrow Pre_{\mt{x}, \mt{u}}(\underline{\mm{X}}_{k-1}^i, \mt{A}, \mt{B}, \mm{U}, \mm{W}, \mm{L})$
  \IF {$LE({\mt{z}}^*, \mt{f}, \mm{Z}_{k}^i) \subseteq  \mm{L} $ }
  \STATE $\mm{C}_{k}^i \leftarrow \{Proj_{\mt{x}}(\mm{Z}_{k}^i) \cap \mm{X}_{\rm safe}\}$
  \STATE \textbf{break}
  \ELSE
  \STATE Select a generator $\mt{g}_{k-1}^{i,j^\ast}$ of $\underline{\mm{X}}_{k-1}^i$ 
  \STATE Split $\underline{\mm{X}}_{k-1}^i$ into $\underline{\mm{X}}_{k-1,1}^{i,j^*}$ and  $\underline{\mm{X}}_{k-1,2}^{i,j^*}$ \COMMENT{Proposition \ref{prop:splitting}}
  \STATE $\mm{C}_{k-1}^i \leftarrow \{\underline{\mm{X}}_{k-1,1}^{i,j^*}, \,  \underline{\mm{X}}_{k-1,2}^{i,j^*}\}$ 
  \STATE $\mm{C}_k^i \leftarrow$ SplittingBRS($\mm{C}_{k-1}^i, \mt{f}, \mm{U}, \mm{W}, \mt{\bar{L}}, \mm{X}_{\rm safe}$)
  \ENDIF
  \ENDFOR
 \RETURN $\{\underline{\mm{X}}_k^\ell\} \leftarrow \bigcup_i \mm{C}_k^i$
 \end{algorithmic}
 \label{alg:split}
 \end{algorithm}

  In line 10, the generator $\mt{g}_{k-1}^{i,j^*}$ is selected as follows. Similar to \cite{althoff2008reachability}, for the $j^{\rm th}$ generator of the set $\underline{\mm{X}}^i_{k-1}$ to split, we compute a performance index $\rho_j$ as follows: 
 \begin{align}
     \rho_j = \max(\mt{L}_1^j / \mt{\bar{L}}) \cdot \max(\mt{L}_2^j / \mt{\bar{L}}),  \label{eq:select_dim}
 \end{align}
where $\mt{L}_1^j$,  $\mt{L}_2^j \in \mathbb{R}^n$ are vectors that define the linearization error bound for sets $\mm{X}_{k,1}^{i,j}$ and $\mm{X}_{k,2}^{i,j}$, respectively. The operations $\max$ and $/$ in Eq. \eqref{eq:select_dim} are element-wise. In line 10, the generator $\mt{g}_{k-1}^{i,j^*}$ with the lowest performance index will be chosen, i.e., 
$j^* = \mathop{\arg\min}_j \rho_j $. 

\begin{rmk}\label{rmk:union}
Note that, while executing SplittingBRS in line 13, the set  $\mm{L} \oplus \mm{W}$ will be subtracted from the two sets $\underline{\mm{X}}_{k-1,1}^{i,j^*}$ and $\underline{\mm{X}}_{k-1,2}^{i,j^*}$, which are obtained via splitting. 
By Lemma \ref{lem:Minkowski}, bullet iii), 
the union of these two Minkowski differences 
is only a subset of (but not necessarily equal to) $\underline{\mm{X}}_{k-1}^i\ominus (\mm{L} \oplus \mm{W})$. 
This means that the splitting procedure introduces more conservatism. However, the overlapping area generated by Proposition \ref{prop:splitting} can ease the conservatism.  
This is because the larger $\underline{\mm{X}}_{k-1,1}^{i,j^*}$ and $\underline{\mm{X}}_{k-1,2}^{i,j^*}$ are, the larger $\mm{C}_k^i$ is. 
\end{rmk}

When doing the splitting, the number of the obtained constrained zonotopes may grow exponentially. Therefore, a sampling algorithm is required to restrict their number. In order to evenly cover the union of these sets, \revise{ we implement the farthest point sampling algorithm \cite{gonzalez1985clustering} based on the geometric centers of the constrained zonotopes' interval closures.}





\section{Examples}
In this section we illustrate our algorithms with several examples. 
\revise{TABLE} \ref{tab:result} summarizes our results with, for each example and method, the system dimension $n$, the iteration steps $k$ and the computing time.
These examples were run on a laptop with a 12th generation Intel CPU and 16 GB of RAM. \revise{Our implementation is in MATLAB R2019a. The zonotope-based method and the HJB method that we use as benchmarks are also in MATLAB.} 
\begin{table}[h]
    \centering
    \begin{tabular}{c|cc|ccc}
    & $n$& $k$ & Splitting& Scaling&HJB  \\
    \hline
    Example 2  & 2 & 100 & N/A & 56.1s & 45.2s\\
     \hline
    Example 3 & 10 & 10 & 226.7s & N/A & Memory error \\
     \hline
    \begin{tabular}{c}
         Example 4:\\
        Convex constraints 
    \end{tabular}
     & 3 & $-$ &          \begin{tabular}{c}
         478.9s \\
         ($k=25$) 
    \end{tabular} & \begin{tabular}{c}
         2821.3s \\
         ($k=400$) 
    \end{tabular}  & Memory error\\
      \hline
        \begin{tabular}{c}
         Example 4: \\
         Nonconvex constraints 
    \end{tabular} & 3 & 20 & 1564.1s & N/A & 4521.6s\\
     \hline
    Example 5 & 10 & 340 & 951.6s & N/A & Memory error \\

    \end{tabular}
    \vspace{-1mm}
    \caption{Computation time for the examples}
    \label{tab:result}
    \vspace{-5mm}
\end{table}
 Note that the splitting method does not apply to Example \ref{eg:brsl} because the system is linear and $\mm{X}_{\rm safe}$ is convex (hence no reason for splitting). We also do not apply the scaling methods to examples with nonconvex $\mm{X}_{\rm safe}$ (Examples~\ref{eg:brs_dubins_split} and \ref{eg:brs_10d_tank}), because its implementation only generates one homotopy class. 

While Algorithm \ref{alg:scaling},  \ref{alg:split} are developed for nonlinear systems, they both reduce to Eq. \eqref{eq:brs_lin} for linear systems (i.e., when $\mt{f}(\mt{x},\mt{u}) = \mt{A}\mt{x}+\mt{B}\mt{u})$. 
The following linear system examples show that less conservative under-approximations  can be obtained using constrained zonotopes instead of zonotopes, because the former
 has a stronger expressive power. 

\begin{figure*}
\begin{minipage}[b]{0.3\textwidth}
   \centering
   \scriptsize
    \begin{tabular}{c|cccc}
    & Exact & \begin{tabular}{c}
         Constrained \\ zonotope
    \end{tabular}   & Zonotope & HJB \\
    \hline
    Volume & 37.079 & 28.343 & 7.810 & 2.817
    \end{tabular}
    \captionof{table}{Example \ref{eg:brsl}, the volume of the BRSs.}
    \label{tab:eg2}
    \centering
    \begin{tabular}{c|cccc}
    &  \begin{tabular}{c} Constrained \\ zonotope \end{tabular} & HJB & Intersection\\
    \hline
    Volume & 0.1608 & 0.3422 & 0.1504 \\
    \begin{tabular}{c} Projection \\ Volume \end{tabular} &0.7070 & 0.7326 & 0.6674
    \end{tabular}
    \captionof{table}{Example \ref{eg:brs_dubins_split}, the volume of the BRSs and their projections.}
    \label{tab:eg4}
\end{minipage}%
\hfill
\begin{minipage}[b]{0.3\textwidth}
  \includegraphics[width=\linewidth]{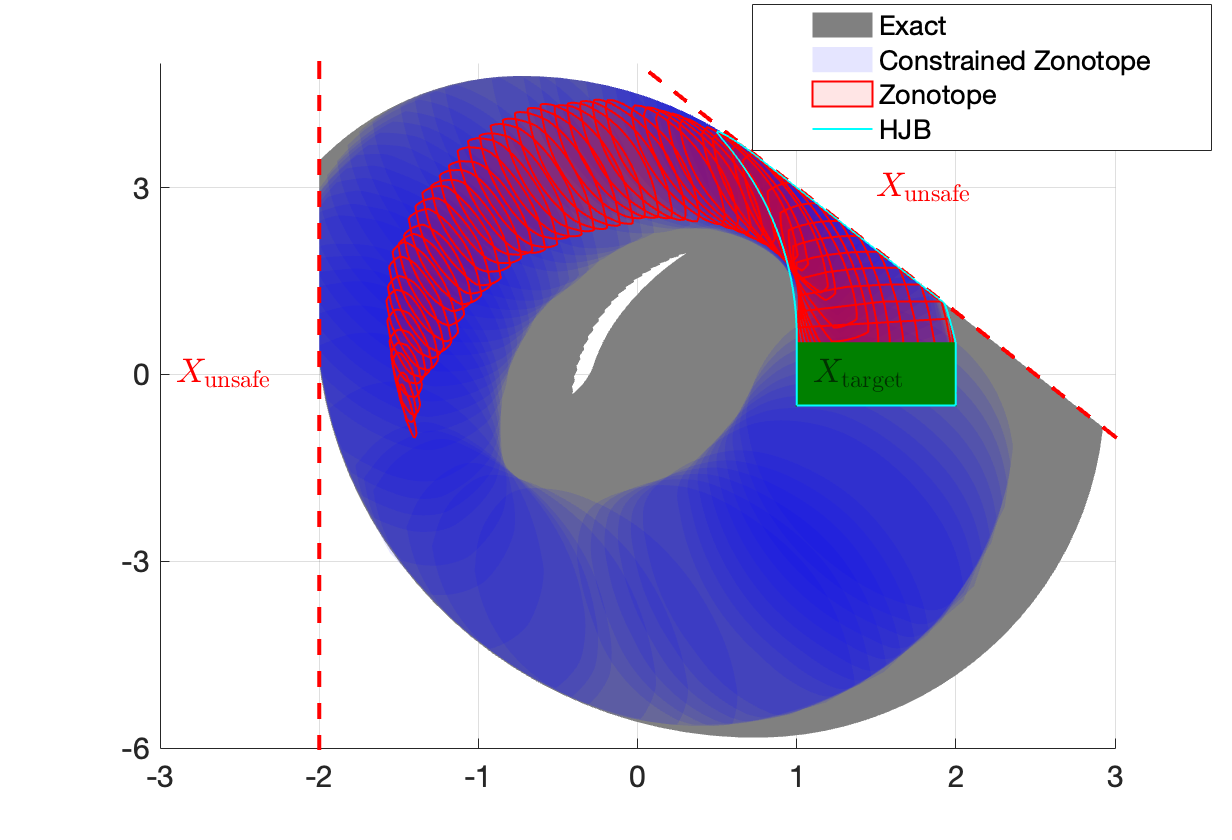}
  \captionof{figure}{Example \ref{eg:brsl}, exact BRSs and their under-approximations.
  }
  \label{fig:eg_brsl}
\end{minipage}%
\hfill 
\begin{minipage}[b]{0.3\textwidth}
  \includegraphics[width=\linewidth]{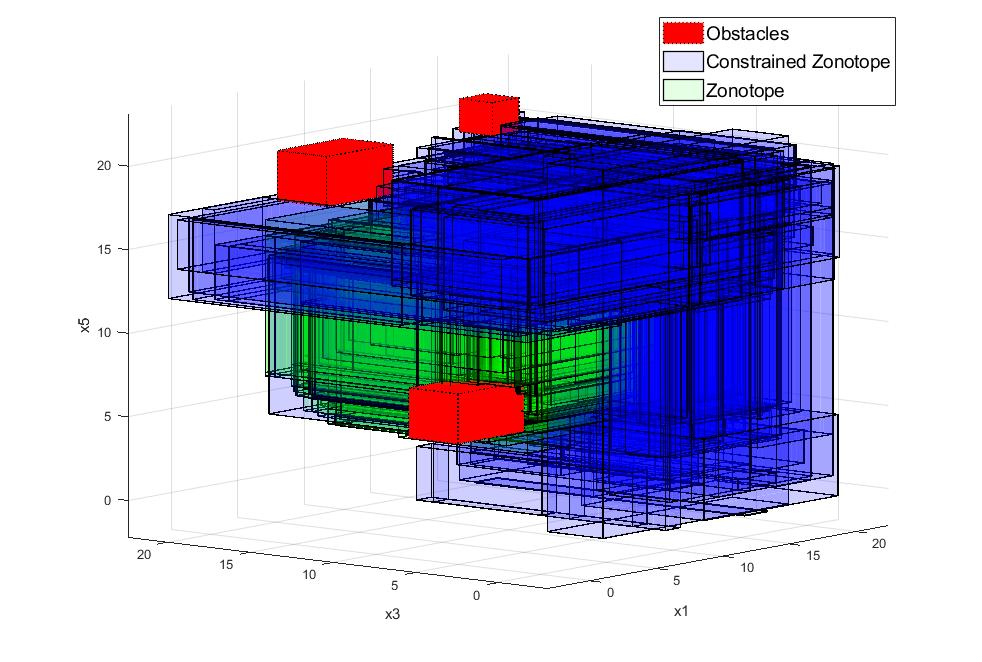}
  \captionof{figure}{Example \ref{eg:brs_10d}, BRSs using the zonotope-based \& the constrained-zonotope-based methods.
  }
  \label{fig:eg_brs_10d}
\end{minipage}
\vspace{-6mm}
\end{figure*}

\begin{figure*}
\begin{minipage}[t]{0.3\textwidth}
  \includegraphics[width=\linewidth]{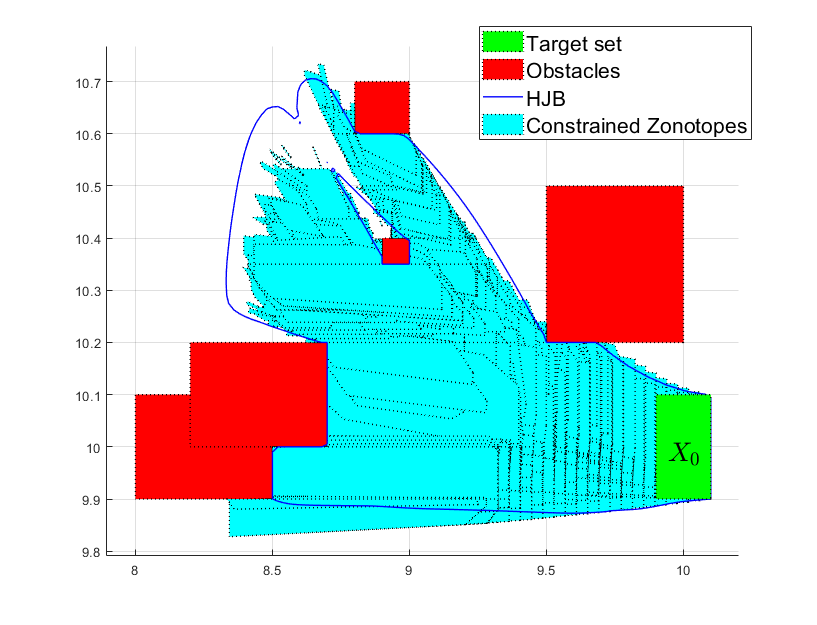}
  \caption{Example \ref{eg:brs_dubins_split}, BRSs of the Dubins car system with obstacles by the splitting method.
  }
  \label{fig:eg_brsnl_split_dubins}
\end{minipage}%
\hfill
\begin{minipage}[t]{0.3\textwidth}
  \includegraphics[width=\linewidth]{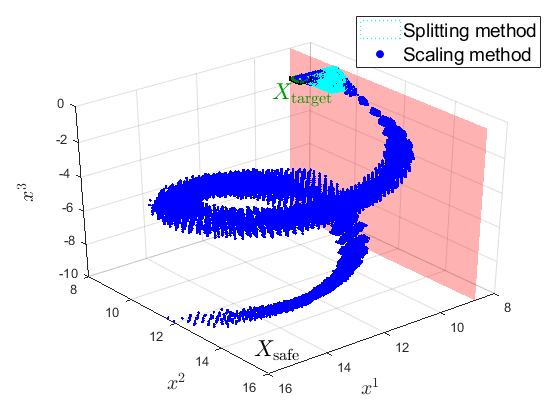}
  \caption{Example \ref{eg:brs_dubins_split}, BRSs of the Dubins car system with convex constraints using the scaling method and the splitting method.
  }
  \label{fig:eg_brsnl_scale}
\end{minipage}%
\hfill 
\begin{minipage}[t]{0.3\textwidth}
  \includegraphics[width=\linewidth]{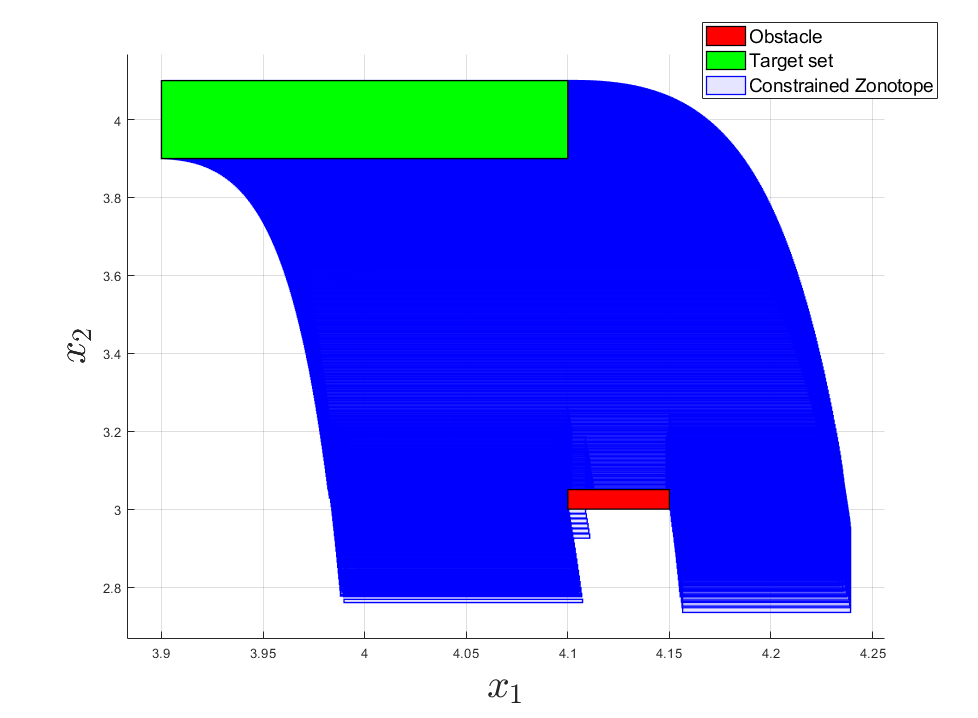}
  \caption{Example \ref{eg:brs_10d_tank}, BRSs for the 10-D tank system.
  }
  \label{fig:eg_brs_tank}
\end{minipage}
\vspace{-7mm}
\end{figure*}

\begin{eg}\label{eg:brsl}
Consider a linear system with the following system matrices and sets: 
\begin{align}
\mt{A} & = 
\left[\hspace{-1mm}
\begin{array}{cc}
0.9962 & 0.02394  \\
-0.1496 & 0.9962
\end{array}
\hspace{-1mm}\right], 
\mt{B} =   
\left[\hspace{-1mm}
\begin{array}{c}
-0.004034 \\
0.08025
\end{array}
\hspace{-1mm}\right],  
\end{align}
$\mm{U} = \li -1.5, 1.5\ri$,  
$ \mm{W} = \langle [0.1997, 0.002396; -0.01498, $ $0.1997], \bfz \rangle$, 
$\mm{X}_0  = \langle \text{diag}([0.5, 0.5]), [1.5; 0] \rangle$ and $
\mm{X}_{\rm safe}  = \{x\in \mathbb{R}^2 \mid [-1,0;2,1]x\leq [2;5]\}$.
Fig. \ref{fig:eg_brsl} shows the exact BRSs $\mm{X}_k$ (gray) for $k=1,2\dots, 100$ and their under-approximations. The constrained zonotopic under-approximations $\underline{\mm{X}}_k$ are in blue. 
As a comparison, we used the method in \cite{yang2022scalable} to compute zonotopic under-approximations (red) and scaled the generators of these zonotopes to satisfy the linear safety constraints.
The latter approach is clearly more conservative (i.e. gives smaller sets). In fact, due to the wrapping effect after hitting the unsafe set $\{\mt{x} \mid [2,1]\mt{x} > 5\}$, the red sets vanish before $k=100$ is reached. 
The main reason of this conservatism is that the true backward reachable set $\mm{X}_k$ becomes asymmetric due to the state constraints. Therefore it is more accurate to approximate $\mm{X}_k$ with a constrained zonotope than with a zonotope. The former is as expressive as polytopes while the latter is restricted to be a centrally symmetric set. 

 In addition, we use the HJB method to compute the BRSs (cyan contour) with the same constraints. The result using HJB is more conservative and stops expanding after $k=20$. The volumes of the BRSs obtained using different methods are approximated using a sample-based method and are shown in TABLE \ref{tab:eg2}. 

\end{eg}



\begin{eg} \label{eg:brs_10d}
    Consider the 10-D system from \cite{yang2022scalable}, discretized with a sampling period of $dt = 0.1$s. Let the disturbance set $\mm{W}$ be so that $w_{\{1,3,5\}} \in \li-0.12, 0.12\ri$, $w_{\{2,4,6\}} \in \li-0.2, 0.2\ri$, $w_{\{7,8,9,10\}} \in \li -0.1, 0.1\ri$ and the control set $\mm{U} \in \li -0.5, 0.5 \ri^3$. Define the target set such that $x_i \in \li 9.5, 10.5 \ri$ for $i \in \{ 1,2,3,4,5,6\}$ and $x_i \in \li 8, 12 \ri$ for $i \in \{ 7,8,9,10\}$.
    
    To avoid potential numerical issues when visualizing 10-D zonotopes or constrained zonotopes, bounding boxes are used to visualize the results. Fig. \ref{fig:eg_brs_10d} shows the 3-D projection of the boxes including the constrained zonotopic under-approximation of BRSs (blue) and zonotopic under-approximation of BRSs (green) for $k = 1,2 \cdots, 10$. Since the system dimension is large, Hamilton-Jacobi method encountered memory error, whereas both zonotope and constrained zonotope-based methods can obtain a result. Here we used the approach in Sec. \ref{subsec:splitting} to avoid obstacles. In this example, the zonotopic representation is more conservative than that based on constrained zonotopes, while the latter being as scalable as the former. That is, constrained zonotopes can also handle high dimensional linear systems, as zonotopes do. 
    
\end{eg} \label{eg:brs_10d_tank}


\begin{eg} \label{eg:brs_dubins_split}
    Consider the following Dubins Car system: 
    \begin{align}
        x_{t+1}^1 &= x_t^1 + u_t^1\cos(x_t^3) \\
        x_{t+1}^2 &= x_t^2 + u_t^1\sin(x_t^3) \\
        x_{t+1}^3 &= x_t^3 + u_t^2
    \end{align}
    Assume that $\mt{x}_t = [x_t^1;x_t^2;x_t^3]$ and $\mt{u}_t = [u_t^1; u_t^2] \in \li 0.04, 0.08\ri \times \li 0, 0.04 \ri$ is the control input. We use the scaling method and the splitting method to compute the BRSs with convex and nonconvex state constraints, respectively. 
    
    Fig. \ref{fig:eg_brsnl_split_dubins} shows the $[x^1;x^2]$-projection of the constrained zonotopic under-approximation $\underline{\mm{X}}_k$ (cyan) of the BRSs, obtained by the splitting method. As a comparison, we also use the HJB method \cite{bansal2017hamilton} to approximate the BRSs (blue contour). To this end, a uniform grid ($201 \times 201 \times 101$) of the state space is used. 
    The BRSs obtained via these two methods both contain states 
    from different homotopy classes in an environment with obstacles. Further, the two methods give BRSs that are similar in sizes but not comparable in the set inclusion sense (Fig \ref{fig:eg_brsnl_split_dubins} \& TABLE \ref{tab:eg4}). 
    In particular, when expanding into the free state space, the HJB method tends to give larger BRSs than the splitting method. 
    However, the splitting method is faster (TABLE \ref{tab:result}). 
    The volumes of the BRSs (and their projections) obtained using both methods are approximated using a sample-based method and are shown in TABLE \ref{tab:eg4}.
 
    Fig. \ref{fig:eg_brsnl_scale} shows $\underline{\mm{X}}_k$ obtained by the scaling method (blue) and the splitting method (cyan).  Here the safe set is a single half-space (specified by the red plane). For small $k$'s, the splitting method finds larger BRSs than the scaling method. 
    However, the splitting method has difficulties to proceed for $k\geq 25$. This is because, in the splitting method, $\underline{\mm{X}}_k$ is represented as a collection of small sets, whose number grows fast without an obstacle ``pruning'' these sets in a convex domain. It is also conservative to 
    Minkowski subtract $\mm{L} \oplus \mm{W}$ from each small set in the collection, and uses the union of the obtained Minkowski-differences to compute $\underline{\mm{X}}_{k+1}$ (see Remark \ref{rmk:union}). On the contrary, the scaling method, which computes one set at each step, does not suffer from these issues and can compute the BRSs in a convex domain for a longer time horizon. 

\end{eg}



\begin{eg} \label{eg:brs_10d_tank}
    Consider a 10-D water tank system with the following dynamics derived from the Torricelli’s law:
    \begin{align}
        x_{t+1}^1 &= x_t^1  + dt\cdot(u - k_2x_t^{10} -k_1\sqrt{2gx_t^1}) \\
        x_{t+1}^i &= x_t^i  + dt\cdot(k_1(\sqrt{2gx_t^{i-1}}-\sqrt{2gx_t^i})), \forall i \neq 1
    \end{align}
    where $x^i$ are the water level of the $i^{\rm th}$ tank, $u \in \li 0.135, 0.145 \ri$ is the inflow, $dt=0.01$, $k_1 = 0.015$, $k_2=0.01$, and $g=9.81$. The target set is $\mm{X}_0 = \li 3.9, 4.1 \ri^{10}$
    
    We apply the splitting method to this example. Figure \ref{fig:eg_brs_tank} shows the 2-D projections of the target set (green), the obstacle (red), and the bounding boxes (blue) that include the obtained constrained zonotopic under-approximations of the BRSs.  For this example, the HJB toolbox reports a memory error due to the large grid size, which is necessary for this 10-D system. We manage to compute the BRSs (with two homotopy classes) in reasonable time (TABLE \ref{tab:result}). This example shows that our method can deal with high-dimensional nonlinear systems with nonconvex state constraints. 
\end{eg}

\section{Conclusion \& Future Work}
In this paper, we developed constrained-zonotope-based methods to under-approximate the BRSs for discrete-time nonlinear systems. 
Our main technical contribution was twofold. 
First, we developed an efficient way to under-approximate the Minkowski difference between a constrained zonotopic minuend and a zonotopic subtrahend, which is a necessary step in the sequential BRS computation. 
Our under-approximation was shown to be exact for minuends with rich enough CG-Reps. 
Secondly, using the developed Minkowski difference computation technique, we proposed two methods, i.e., the scaling method and the splitting method, for BRS computation. 
Experiments showed that these constrained-zonotope-based methods were less conservative than those using zonotopes, and were more scalable than the HJB method. 

The exactness result in Sec. \ref{sec:rep} suggests that, for constrained zonotopes, there is a trade-off between the computational  complexity and the accuracy of set operations.  
This trade-off may be better understood via 
a systematic conversion between the different CG-Reps of a constrained zonotope. This conversion may be used, e.g., to incrementally  enrich the CG-Rep of a constrained zonotopic minuend and improve our two-step approach's accuracy. 
We will explore this in the future.

\balance
\bibliographystyle{IEEEtran}
\bibliography{main}


\end{document}